\documentclass[10pt]{article}

\usepackage{amsmath}
\usepackage{amssymb}

\usepackage{cite}

\usepackage{hyperref}

\usepackage{lineno}

\usepackage{microtype}
\DisableLigatures[f]{encoding = *, family = * }

\usepackage[labelfont=bf,labelsep=period,justification=raggedright]{caption}

\makeatletter
\renewcommand{\@biblabel}[1]{\quad#1.}
\makeatother

\date{}

\usepackage{graphicx}

\usepackage{amsthm}
\theoremstyle{plain}
\newtheorem{result}{Result}
\newtheorem{assum}{Assumption}
\newtheorem{theorem}{Theorem}

\newcommand{\M}{\mathrm{M}}
\newcommand{\R}{\mathrm{R}}
\newcommand{\vs}{\mathbf{s}}
\newcommand{\va}{\mathbf{a}}
\newcommand{\vm}{\mathbf{m}}
\newcommand{\cM}{\mathcal{M}}
\newcommand{\cA}{\mathcal{A}}

\DeclareMathOperator{\E}{E}
\DeclareMathOperator{\Prob}{P}

\begin{document}

\begin{flushleft}
{\Large
\textbf{The molecular clock of neutral evolution can be accelerated or slowed by asymmetric spatial structure}
}
\\
Benjamin Allen$^{1,2,3,\ast}$, 
Christine Sample$^{1}$, 
Yulia A.~Dementieva$^{1}$,
Ruben C.~Medeiros$^{1}$,
Christopher Paoletti$^{1}$,
Martin A.~Nowak$^{2,4}$
\\
\bf{1} Department of Mathematics, Emmanuel College, Boston, MA, USA
\\
\bf{2} Program for Evolutionary Dynamics, Harvard University, Cambridge, MA, USA
\\
\bf{3} Center for Mathematical Sciences and Applications, Harvard University, Cambridge, MA, USA
\\
\bf{4} Department of Mathematics, Department of Organismic and Evolutionary Biology, Harvard University, Cambridge, MA, USA\\
$\ast$ E-mail: benjcallen@gmail.com
\end{flushleft}

\section*{Abstract}
Over time, a population acquires neutral genetic substitutions as a consequence of random drift. A famous result in population genetics asserts that the rate, $K$, at which these substitutions accumulate in the population coincides with the mutation rate, $u$, at which they arise in individuals: $K=u$.  This identity enables genetic sequence data to be used as a ``molecular clock" to estimate the timing of evolutionary events.  While the molecular clock is known to be perturbed by selection, it is thought that $K=u$ holds very generally for neutral evolution. Here we show that asymmetric spatial population structure can alter the molecular clock rate for neutral mutations, leading to either $K<u$ or $K>u$.  Deviations from $K=u$ occur because mutations arise unequally at different sites and have different probabilities of fixation depending on where they arise.  If birth rates are uniform across sites, then $K \leq u$. In general, $K$ can take any value between 0 and $Nu$. Our model can be applied to a variety of population structures. In one example, we investigate the accumulation of genetic mutations in the small intestine. In another application, we analyze over 900 Twitter networks to study the effect of network topology on the fixation of neutral innovations in social evolution.

\section*{Author Summary}
Evolution is driven by genetic mutations.  While some mutations affect an organism's ability to survive and reproduce, most are neutral and have no effect.  Neutral mutations play an important role in the study of evolution because they generally accrue at a consistent rate over time.  This result, first discovered almost 50 years ago, allows neutral mutations to be used as a ``molecular clock" to estimate, for example, how long ago humans diverged from chimpanzees and bonobos.  We used mathematical modeling to study how the rates of these molecular clocks are affected by the spatial arrangement of a population in its habitat.  We find that asymmetry in this spatial structure can either slow down or speed up the rate at which neutral mutations accrue.  This effect could potentially skew our estimates of past events from genetic data.  It also has implications for a number of other fields.  For example, we show that the architecture of intestinal tissue can limit the rate of genetic substitutions leading to cancer.  We also show that the structure of social networks affects the rate at which new ideas replace old ones.  Surprisingly, we find that most Twitter networks slow down the rate of idea replacement.

\section*{Introduction}

A half-century ago, Zuckerkandl and Pauling \cite{Zuckerkandl1965Divergence} discovered that amino acid substitutions often occur with sufficient regularity as to constitute a ``molecular clock".  Theoretical support for this observation was provided by Kimura \cite{Kimura1968Rate}, who argued that observed rates of amino acid substitution could only be explained if the majority of substitutions are selectively neutral. Under simple models of evolution, a neutral mutation has probability $1/N$ of becoming fixed in a haploid population of size $N$. It follows that the rate $K$ of neutral substitution per generation---given by the product of the population size $N$, the mutation probability $u$ per reproduction, and the fixation probability $\rho$---is simply equal to $u$.  (A similar cancellation occurs in diploids, leading again to $K=u$.) In other words, for any neutral genetic marker, the rate of substitution at the population level equals the rate of mutation at the individual level \cite{Kimura1968Rate}.  

A number of factors can alter the rate of neutral substitution \cite{Ayala1999Molecular,bromham2003modern}, including selection, changes in population demography over time, or mutation rates that vary systematically by demographic classes \cite{pollak1982rate,balloux2012substitution} or sex \cite{lehmann2014stochastic}. The extent to which these factors compromise the applicability of the molecular clock hypothesis to biological sequence data has been intensely debated \cite{Ohta1971Constancy,Kimura1974Principles,Langley1974Examination,Kimura1984Neutral,Ayala1986Virtues,Li1987Evaluation,Ayala1999Molecular,bromham2003modern,kumar2005molecular}.  

However, it is generally thought that spatial structure alone (without spatial variation in the mutation rate\cite{lehmann2014stochastic}) cannot alter the rate of neutral substitution.  This consensus is based on analyses \cite{Maruyama1970Fixation,Lande1979Effective,Slatkin1981Fixation,Nagylaki1982Invariance,tachida1991fixation,Barton1993Probability,roze2003selection,Whitlock2003Fixation,Lieberman2005Graphs,Broom2008Analysis,Patwa2008Fixation,Masuda2009Evolutionary,Houchmandzadeh2011Fixation,shakarian2012review,voorhees2013birth,constable2014population,monk2014martingales} of a variety of models of spatially structured populations.  Each of these analyses found that the fixation probability of a neutral mutation is $1/N$, and thus the rate of neutral substitution $K=N u \rho$ again simplifies to $u$.

Here we show that the absence of spatial effects on the rate of neutral substitution in these models does not represent a general principle of evolution.  Rather, it is an artifact of two common modeling assumptions: (i) all spatial locations are equivalent under symmetry, and (ii) mutations are equally likely to arise in each location.  While assumption (i) is relaxed in a number of models, assumption (ii) is made almost universally.  Either of these assumptions alone is sufficient to guarantee $\rho=1/N$ and $K=u$, as we will show.  However, neither of these assumptions is necessarily satisfied in a biological population.  In particular, assumption (ii) is violated if some spatial locations experience more turnover (births and deaths) than others. Since each birth provides an independent opportunity for mutation, the rate at which new mutations appear at a location is proportional to its turnover rate \cite{Allen2012Success}.  Thus, even with a constant probability of mutation per birth, mutations may appear with different frequency at different locations.  If, in addition, fixation probability depends on a mutant's initial location, the rate of neutral substitution is altered.  

Our goal is to identify conditions under which the molecular clock rate is maintained ($K=u$), accelerated ($K>u$), or slowed ($K < u$) by spatial population structure. Our main results are as follows (see also Figure \ref{fig:Venn}):

\begin{figure}
\begin{center}
\includegraphics{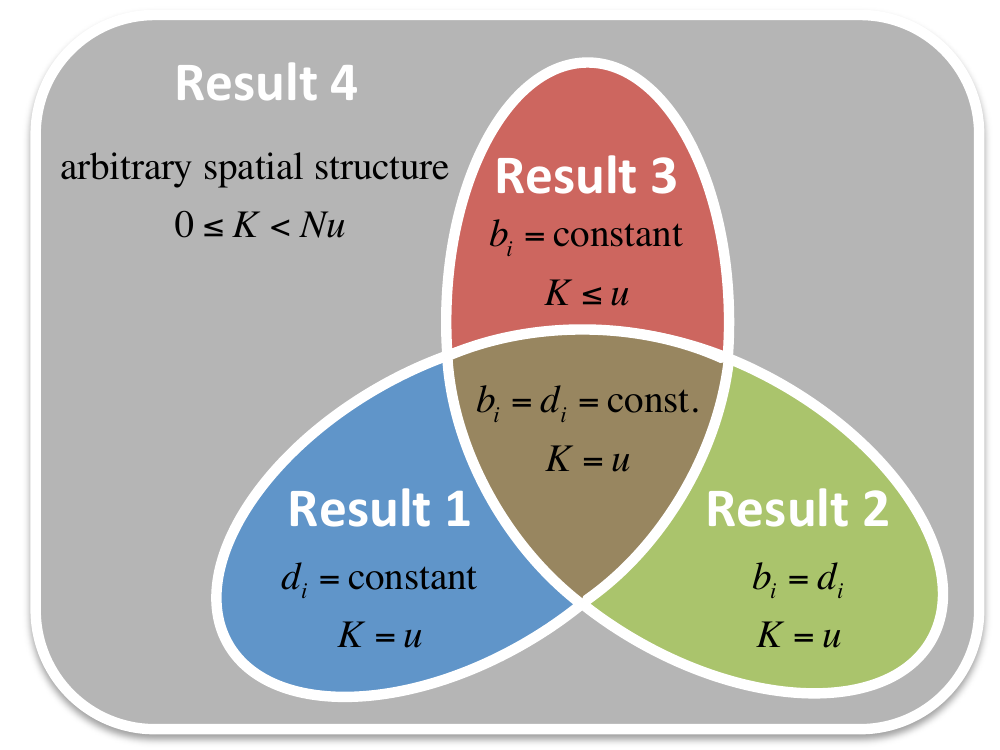}
\caption{\textbf{How spatial structure effects the molecular clock rate $K$.}  Relative to the rate in a well-mixed population ($K=u$), spatial structure can either accelerate ($K>u$) or slow ($K<u$) the accumulation of neutral substitutions, depending on how birth rates $b_i$ and death rates $d_i$ vary across sites.  The rate is unchanged from that of a well-mixed population ($K=u$) if either death rates are uniform across sites (Result 1), or the birth and death rates are equal at each site (Result 2).  Almost all previous studies of neutral drift in spatially structured populations fall into one of these two categories; thus the effects of spatial structure on the molecular clock rate are unappreciated.  We show that, in general, $K$ can take any non-negative value less than $Nu$ (Result 4).  If one adds the constraint that the birth rate is the same at each site, then the molecular clock rate cannot exceed that of a well-mixed population ($K \leq u$; Result 3).}
\label{fig:Venn}
\end{center}
\end{figure}

\begin{itemize}
\item[]{\bf Result 1} If the death rate is constant over sites, the molecular clock rate is identical to that of a well-mixed population: $K=u$.

\item[]{\bf Result 2} If the birth rate equals the death rate at each site, the molecular clock rate is identical to that of a well-mixed population: $K=u$.

\item[]{\bf Result 3}  If the birth rate is constant over sites, the molecular clock rate is less than or equal to that of a well-mixed population: $K \leq u$.

\item[]{\bf Result 4}  In general (with no constraints on birth or death rates), the molecular clock rate $K$ can take any non-negative value less than $Nu$.
\end{itemize}

\section*{Results}

\subsection*{Obtaining the neutral substitution rate}

Our investigations apply to a class of evolutionary models (formally described in the Materials and Methods) in which reproduction is asexual and the population size and spatial structure are fixed.  Specifically, there are a fixed number of sites, indexed $i=1, \ldots, N$.  Each site is always occupied by a single individual.  At each time-step, a replacement event occurs, meaning that the occupants of some sites are replaced by the offspring of others.  The replacement event is chosen according to a probability distribution---called the replacement rule---which depends on the model in question.  Since we consider only neutral mutations that have no phenotypic effect, the probabilities of replacement events are the same for every population state.

This class includes many established evolutionary models.  One important subclass is spatial Moran processes \cite{liggett1985particle,nowak1994more,Lieberman2005Graphs, antal2006evolutionary}, in which exactly one reproduction occurs each time-step. This class also includes spatial Wright-Fisher processes, in which the entire population is replaced each time-step \cite{NowakMay,taylor2011inclusive}.  In general, any subset of individuals may be replaced by the offspring of any other subset in a given time-step.  

For a given model in this class, we let $e_{ij}$ denote the probability that the occupant of site $j$ is replaced by the offspring of site $i$ in a single time-step. Thus the expected number of offspring of site $i$ over a single time-step is $b_i = \sum_{j=1}^N e_{ij}$.  The probability that node $i$ dies (i.e., is replaced) in a time-step is $d_i= \sum_{j=1}^N e_{ji}$.  The death rate $d_i$ can also be regarded as the rate of turnover at site $i$.  The total expected number of offspring per time-step is denoted $B = \sum_{i=1}^N b_i =\sum_{i=1}^N d_i = \sum_{i,j} e_{ij}$.
We define a generation to be $N/ B$ time-steps, so that, on average, each site is replaced once per generation.

We use this framework to study the fate of a single neutral mutation, as it arises and either disappears or becomes fixed.  The probability of fixation depends on the spatial structure and the initial mutant's location. We let $\rho_i$ denote the probability that a new mutation arising at site $i$ becomes fixed. ($\rho_i$ can also be understood as the reproductive value of site $i$ \cite{maciejewski2014reproductive}.)  We show in the Materials and Methods that the fixation probabilities $\rho_i$ are the unique solution to the system of equations
\begin{equation}
\label{eq:rhorecur}
\begin{cases}
\displaystyle d_i \rho_i  = \sum_{j=1}^N e_{ij} \rho_j & \text{for $i=1, \ldots, N$},\\
\displaystyle \sum_{i=1}^N \rho_i  = 1.
\end{cases}
\end{equation}
The identity $\sum_{i=1}^N \rho_i  = 1$ arises because $\rho_i$ equals the probability that the current occupant of site $i$ will become the eventual ancestor of the population, which is true for exactly one of the $N$ sites.  

To determine the overall rate of substitution, we must take into account the likelihood of mutations arising at each site.  The rate at which mutations arise at site $i$ is proportional to the turnover rate $d_i$, because each new offspring provides an independent chance of mutation.  Specifically, if mutation occurs at rate $u \ll 1$ per reproduction, then new mutations arise at site $i$ at rate $N u d_i/B$ per generation \cite{Allen2012Success}.  Thus the fraction of mutations that arise at site $i$ is $d_i/B$.  

The overall fixation probability $\rho$ of new mutations, taking into account all possible initial sites, is therefore
\begin{equation}
\label{eq:rhodef}
\rho=  \frac{1}{B}\sum_{i=1}^N d_i \rho_i.
\end{equation}
The molecular clock rate $K$ is obtained by multiplying the fixation probability $\rho$ by the total rate of mutation per generation:
\begin{equation}
\label{eq:Kdef}
K = N u \rho = \frac{Nu}{B}\sum_{i=1}^N d_i \rho_i.
\end{equation}
The units of $K$ are substitutions per generation.  Alternatively, the molecular clock can be expressed in units of substitutions per time-step, in which case the fomula is $\tilde{K} = Bu\rho = u \sum_{i=1}^N d_i \rho_i$. 

\subsection*{Effects of spatial structure}

How does spatial structure affect the rate of neutral substitution? In a well-mixed population, each individual's offspring is equally likely to replace each other individual, meaning that $e_{ij}$ is constant over all $i$ and $j$ (Fig.~\ref{fig:symmetric}a).  In this case, the unique solution to Eq.~\eqref{eq:rhorecur} is $\rho_i = 1/N$ for all $i$, and we recover Kimura's \cite{Kimura1968Rate} result $K=Nu(1/N)=u$.  Moreover, if each site is equivalent under symmetry, as in Figure \ref{fig:symmetric}b, this symmetry implies that $\rho_i = 1/N$ for all $i$ and $K=u$ as in the well-mixed case.  

\begin{figure}
\begin{center}
\begin{tabular}{cc}
\includegraphics{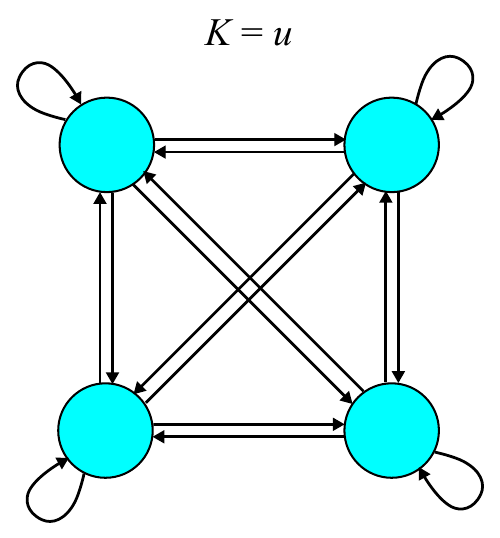} & \includegraphics{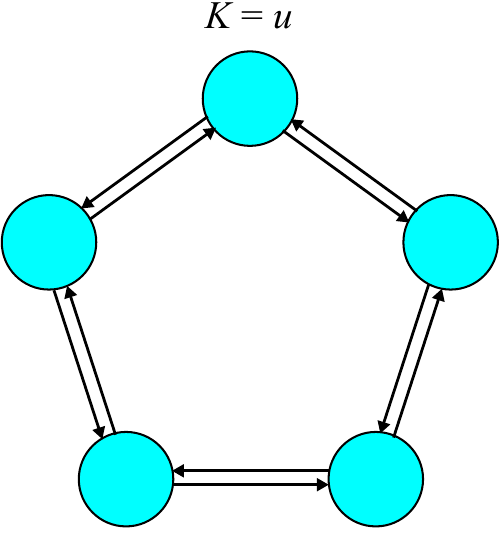}\\
(a) & (b)
\end{tabular}
\caption{\textbf{For spatial structures with symmetry, $K=u$.} (a) For a well-mixed population, represented by a complete graph with uniform edge weights, a neutral mutation has a $1/N$ chance of fixation, where $N$ is the population size.  It follows that the rate $K$ of neutral substitution in the population equals the rate $u$ of neutral mutation in individuals.  (b) The same result holds for spatial structures in which each site is \emph{a priori} identical, such as the cycle with uniform edge weights.}
\label{fig:symmetric}
\end{center}
\end{figure}

However, asymmetric spatial structure can lead to faster ($K>u$) or slower ($K<u$) molecular clock rates than a well-mixed population, as shown in Figure \ref{fig:examples}.  These results led us to seek general conditions for when spatial structure lead to faster, slower, or the same molecular clock rates as a well-mixed population.  

\begin{figure}
\begin{center}
\begin{tabular}{cc}
\includegraphics{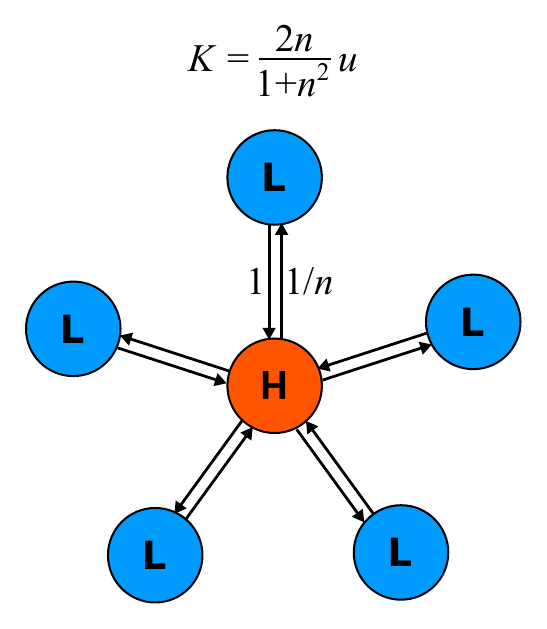} & \includegraphics{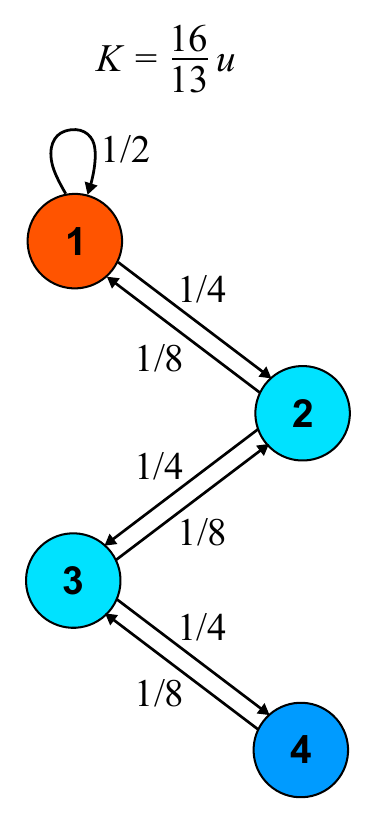}\\
(a) & (b)
\end{tabular}
\caption{\textbf{Asymmetric spatial structure affects the rate of neutral substitution.}  This is because the frequency of mutations and the probability of fixation differ across sites.  Turnover rates are indicated by coloration, with red corresponding to frequent turnover and consequently frequent mutation.  (a) A star consists of a hub and $n$ leaves, so that the population size is $N=n+1$. Edge weights are chosen so that the birth rates are uniform ($b_i=1$ for all $i$). Solving Eq.~\eqref{eq:rhorecur}, we obtain site-specific fixation probabilities of $\rho_\mathrm{H}=1/(1+n^2)$ and $\rho_\mathrm{L}=n/(1+n^2)$ for the hub and each leaf, respectively.  From Eq.~\eqref{eq:Kdef}, the molecular clock rate is $K = \frac{2n}{1+n^2} u$, which equals $u$ for $n=1$ and is less than $u$ for $n\geq 2$.  Thus the star structure slows down the rate of neutral substitution, in accordance with Result 3.  Intuitively, the slowdown occurs because mutations are more likely to arise at the hub, where their chances of fixation are reduced.  (b) A one-dimensional population with self-replacement only in site 1.  Solving Eq.~\eqref{eq:rhorecur} we find $\rho_1=\frac{8}{15}$ and $\rho_2=\frac{4}{15}$, $\rho_3=\frac{2}{15}$ and $\rho_4=\frac{1}{15}$. (The powers of two arise because there is twice as much gene flow in one direction as the other). From Eq.~\eqref{eq:Kdef}, the molecular clock rate is $K=\frac{16}{13}u>u$, thus the molecular clock is accelerated in this case.}
\label{fig:examples}
\end{center}
\end{figure}

We first find
\begin{result}
If the death rates $d_i$ are constant over all sites $i=1, \ldots, N$, then $\rho=1/N$, and consequently $K=u$.  
\end{result}
Thus the molecular clock rate is unaffected by spatial structure if each site is replaced at the same rate (Fig.~\ref{fig:Result12fig}a). This result can be seen by noting that if the $d_i$ are constant over $i$, then since $\sum_{i=1}^N d_i = B$, it follows that $d_i=B/N$ for each $i$.  Substituting in Eq.~\eqref{eq:rhodef} yields $\rho=1/N$.  

\begin{figure}
\begin{center}
\begin{tabular}{cc}
\includegraphics{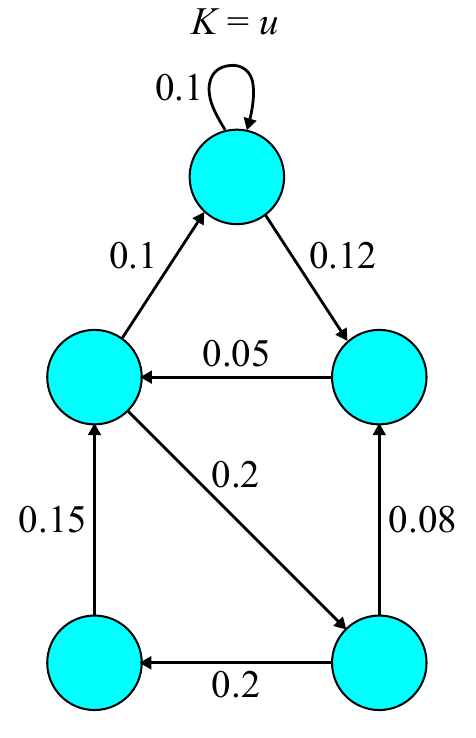} & \includegraphics{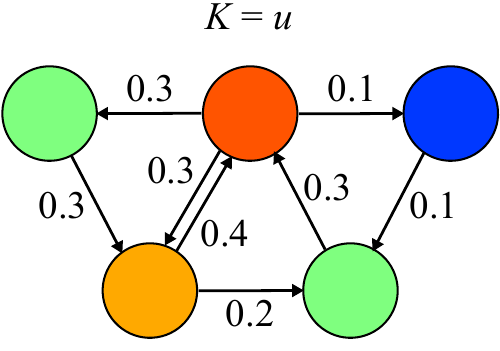}\\
(a) & (b)
\end{tabular}
\caption{\textbf{Results 1 and 2 give conditions leading to $K=u$.} (a) Our Result 1 states that the molecular clock has the same rate as in a well-mixed population, $K=u$, if the rate of turnover $d_i$ is uniform across sites, as in this example ($d_i=0.2$ for all $i$).  (b) Result 2 asserts that $\rho_i=1/N$ for all $i$---again implying $K=u$---if and only if each site has birth rate equal to death rate, $b_i=d_i$ for all $i$, as in this example.  Nodes are colored according to their rates of turnover $d_i$.}
\label{fig:Result12fig}
\end{center}
\end{figure}

Another condition leading to $K=u$ is the following:
\begin{result}
If the birth rate equals the death rate at each site ($b_i = d_i$ for all $i=1, \ldots, N$), then $\rho=1/N$, and consequently $K=u$.  Moreover, $b_i = d_i$ for all $i=1, \ldots, N$ if and only if the fixation probability is the same from each site ($\rho_i=1/N$ for all $i=1, \ldots, N$).
\end{result}
Thus if births and deaths are balanced at each site, then all sites provide an equal chance for mutant fixation (Fig.~\ref{fig:Result12fig}b).  In this case the molecular clock is again unchanged from the baseline value. In particular, if dispersal is symmetric in the sense that $e_{ij}=e_{ji}$ for all $i$ and $j$ then $K=u$.  Result 2 can be obtained by substituting $\rho_i=1/N$ for all $i$ into Eq.~\eqref{eq:rhorecur} and simplifying to obtain $b_i = d_i$ for all $i$ (details in Methods).  Alternatively, Result 2 can be obtained as a corollary to the Circulation Theorem of Lieberman et al.~\cite{Lieberman2005Graphs}.

Our third result reveals a ``speed limit" to neutral evolution in the case of constant birth rates:
\begin{result}
If the birth rates $b_i$ are constant over all sites $i=1, \ldots, N$, then $\rho \leq 1/N$, and consequently $K \leq u$, with equality if and only if the death rates $d_i$ are also constant over sites.
\end{result}
In other words, a combination of uniform birth rates and nonuniform death rates slows down the molecular clock.  An instance of this slowdown in shown in Figure \ref{fig:examples}a.  Intuitively, the sites at which mutations occur most frequenly are those with high death rates $d_i$; because of these high death rates, these sites on the whole provide a reduced chance of fixation.  The proof of this result, however, is considerably more intricate than this intuition would suggest (see Methods).

Finally, we investigate the full range of possible values for $K$ with no constraints on birth and death rates.  We find the following:
\begin{result}
For arbitrary spatial population structure (no constraints on $e_{ij}$) the fixation probability can take any value $0 \leq \rho < 1$, and consequently, the molecular clock can take any rate $0 \leq K < Nu$.
\end{result}
This result is especially surprising, in that it implies that the probability of fixation of a new mutation can come arbitrarily close to unity.  Result 4 can be proven by considering the hypothetical spatial structure illustrated in Figure \ref{fig:noconstraint}.  Any non-negative value of $\rho$ less than 1 can be obtained by an appropriate choice of parameters (details in Methods).

\begin{figure}
\begin{center}
\includegraphics{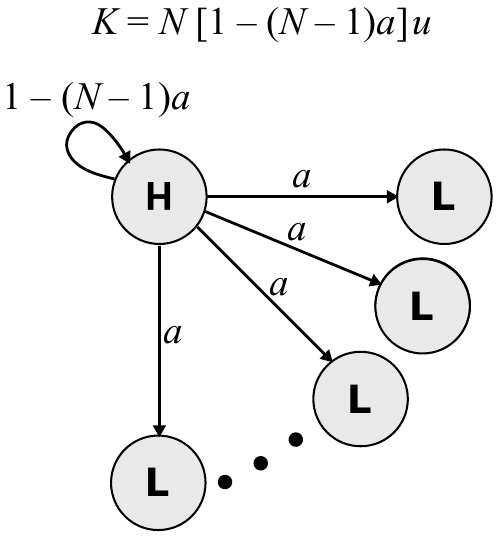}
\caption{\textbf{Result 4 shows that $K$ can achieve any value $0 \leq K < Nu$.}  This is proven by considering a population structure with unidirectional gene flow from a hub (H) to $N-1$ leaves (L).  Fixation is guaranteed for mutations arising in the hub ($\rho_\mathrm{H}=1$) and impossible for those arising in leaves ($\rho_\mathrm{L}=0$).  The overall fixation probability is equal by Eq.~\eqref{eq:rhodef} to the rate of turnover at the hub: $\rho= d_\mathrm{H} = 1-(N-1)a$.  The molecular clock rate is therefore $K=Nu\rho = N[1-(N-1)a]u$.  It follows that $K>u$ if and only if $a <1/N$. Intuitively, the molecular clock is accelerated if the hub experiences more turnover (and hence more mutations) than the other sites.  Any value of $\rho$ greater than or equal to 0 and less than 1 can be achieved through a corresponding positive choice of $a$  less than or equal to $1/(N-1)$.  For $a=1/(N-1)$ we have $K=0$, because mutations arise only at the leaves where there is no chance of fixation.  At the opposite extreme, in the limit $a \to 0$, we have $K \to Nu$.}
\label{fig:noconstraint}
\end{center}
\end{figure}

\subsection*{Application to upstream-downstream populations}

To illustrate the effects of asymmetric dispersal on the molecular clock, we consider a hypothetical population with two subpopulations, labeled ``upstream" and ``downstream" (Figure \ref{fig:updown}). The sizes of these subpopulations are labeled $N_\uparrow$ and $N_\downarrow$, respectively.  Each subpopulation is well-mixed, with replacement probabilities $e_\uparrow$ for each pair of upstream sites and $e_\downarrow$ for each pair of downstream sites.  Dispersal between the subpopulations is represented by the replacement probabilities $e_\rightarrow$ from each upstream site to each downstream site, and $e_\leftarrow$ from each downstream site to each upstream site.  We assume there is net gene flow downstream, so that $e_\rightarrow > e_\leftarrow$.  

\begin{figure}
\begin{center}
\includegraphics{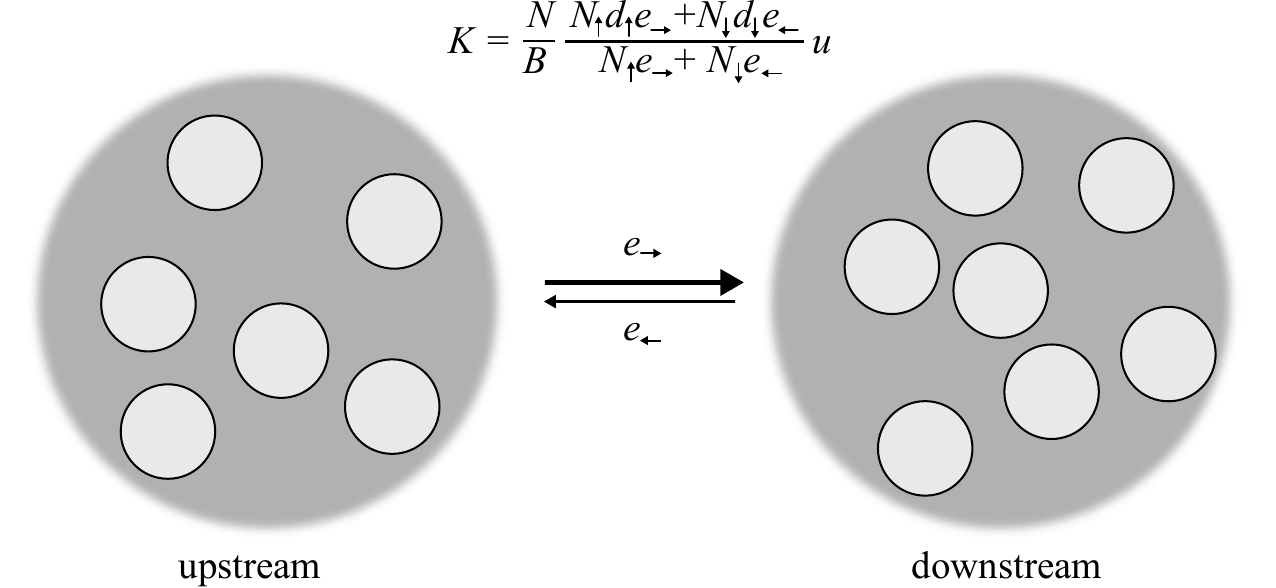}
\caption{\textbf{A model of a population divided into upstream and downstream subpopulations.}  Each subpopulation is well-mixed.  The replacement probability $e_{ij}$ equals $e_\uparrow$ if sites $i$ and $j$ are both upstream, $e_\downarrow$ if $i$ and $j$ are both downstream, $e_\rightarrow$ if $i$ is upstream and $j$ is downstream, and $e_\leftarrow$ if $i$ is downstream and $j$ is upstream.  We suppose there is net gene flow downstream, so that $e_\rightarrow > e_\leftarrow$.  We find that the molecular clock is accelerated, relative to the well-mixed case, if and only if the upstream subpopulation experiences more turnover than the downstream subpopulation: $K>u$ if and only if $d_\uparrow > d_\downarrow$.}
\label{fig:updown}
\end{center}
\end{figure}

Solving Eq.~\eqref{eq:rhorecur}, we find that the fixation probabilities from each upstream site and each downstream site, respectively, are 
\begin{equation}
\begin{split}
\label{eq:rhoupdown}
\rho_{\uparrow} = \frac{e_{\rightarrow}}{ N_{\uparrow}e_{\rightarrow}+ N_{\downarrow}e_{\leftarrow}},\\
\rho_{\downarrow} = \frac{e_{\leftarrow}}{N_{\uparrow}e_{\rightarrow} + N_{\downarrow}e_{\leftarrow}}.
\end{split}
\end{equation}
These fixation probabilities were previously discovered for a different model of a subdivided population \cite{Lundy1998eco}.  Substituting these fixation probabilities into Eq.~\eqref{eq:Kdef} yields the molecular clock rate:
\begin{equation}
K = \frac{N}{B} \,
\frac{N_\uparrow d_\uparrow e_\rightarrow + N_\downarrow d_\downarrow e_\leftarrow}
{N_{\uparrow}e_{\rightarrow} + N_{\downarrow}e_{\leftarrow}} \, u.
\end{equation}
Above, $d_\uparrow$ and $d_\downarrow$ are the turnover rates in the upstream and downstream populations, respectively, and $B = N_{\uparrow}d_{\uparrow}+N_{\downarrow}d_{\downarrow}$ is the total birth rate per time-step.  In Methods, we show that $K>u$ if and only if $d_\uparrow > d_\downarrow$; that is, the molecular clock is accelerated if and only if there is more turnover in the upstream population than in the downstream population.  

\subsection*{Application to epithelial cell populations}

Our results are also applicable to somatic evolution in self-renewing cell populations, such as the crypt-like structures of the intestine.  Novel labeling techniques have revealed that neutral mutations accumulate in intestinal crypts at a constant rate over time \cite{kozar2013continuous}.  The cell population in each crypt is maintained by a small number of stem cells that reside at the crypt bottom and continuously replace each other in stochastic manner (Figure \ref{fig:crypt}; \cite{lopez2010intestinal,snippert2010intestinal,vermeulen2014review}).  We focus on the proximal small intestine in mice, for which recent studies \cite{kozar2013continuous,vermeulen2013defining} suggest there are $\sim 5$ active stem cells per crypt, each replaced $\sim 0.1$ times per day by one of its two neighbors. In our framework, this corresponds to a cycle-structured population of size 5 with replacement rates $0.05/\mathrm{day}$ between neighbors, so that $d_i=0.1/\mathrm{day}$ for all $i$.

\begin{figure}
\begin{center}
\includegraphics{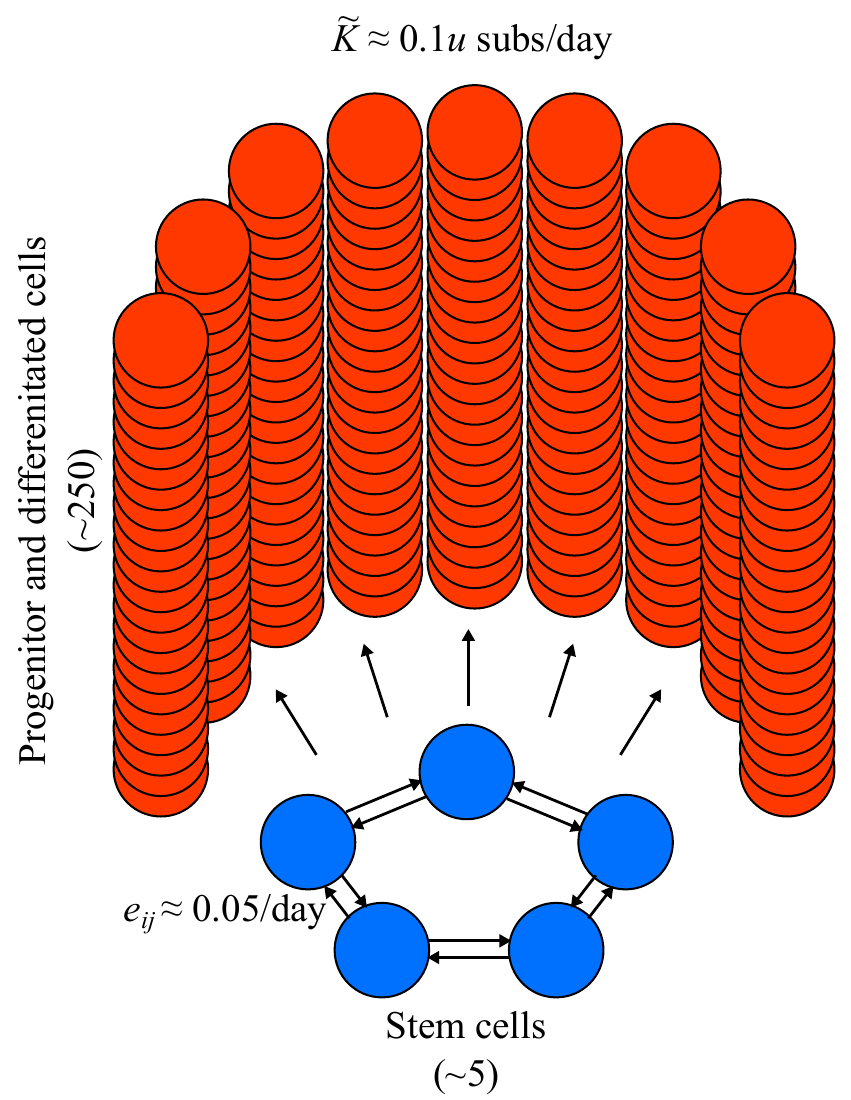}
\caption{\textbf{A simple model of cell replacement structure in epithelial crypts of the small intestine, based on results of \cite{kozar2013continuous} and \cite{vermeulen2013defining}.}  A small number of stem cells ($N_\uparrow \sim 5$) residing at the bottom of the intestinal crypt and are replaced at rate $d_{\uparrow} \sim 0.1$ per stem cell per day.  Empirical results \cite{kozar2013continuous,vermeulen2013defining} suggest a cycle structure for stem cells. To achieve the correct replacement rate we set $e_{ij} = 0.05 / \mathrm{day}$ for each neighboring pair.   Stem cells in an individual crypt replace a much larger number of progenitor and differentiated cells ($\sim 250$; \cite{barker2012identifying}).  These downstream progenitor and differentiated cells are replaced about every day \cite{barker2012identifying}.  The hierarchical organization of  intestinal crypts, combined with the low turnover rate of stem cells, limits the rate of neutral genetic substitutions ($\tilde{K} \approx 0.1u$ substitutions per day), since only mutations that arise in stem cells can fix. }
\label{fig:crypt}
\end{center}
\end{figure}

Only mutations that arise in stem cells can become fixed within a crypt; thus we need only consider the fixation probabilities and turnover rates among stem cells.  By symmetry among the stem cells, $\rho_i = 1/5$ for each of the five stem cell sites.  The molecular clock rate is therefore $\tilde{K} = u \sum_{i=1}^5 d_i \rho_i = 0.1 u$ substitutions per day.  This accords with the empirical finding that, for a neutral genetic marker with mutation rate $u \approx 1.1\times10^{-4}$, substitutions accumulate at a rate $\tilde{K} \approx 1.1 \times 10^{-5}$ per crypt per day \cite{kozar2013continuous}.

Does crypt architecture limit the rate of genetic change in intestinal tissue?  Intestinal crypts in mice contain $\sim 250$ cells and replace all their cells about once per day \cite{barker2012identifying}.  If each crypt were a well-mixed population, the molecular clock rate would be $\tilde{K}=Bu/N \approx u$ substitutions per day.  Thus the asymmetric structure of these epithelial crypts slows the rate of neutral genetic substitution tenfold.

\subsection*{Application to the spread of innovations}

Our results can also be applied to innovations that spread by imitation on social networks.  In this setting, a mutation corresponds to an innovation that could potentially replace an established convention.  Neutrality means that the innovation is equally likely to be imitated as the established convention.  

To investigate whether human social networks promote or hinder the fixation of selectively neutral innovations, we analyzed 973 Twitter networks from the Stanford Large Network Dataset Collection \cite{leskovec2011stanford}.  Each of these ``ego networks" represents follower relationships among those followed by a single ``ego" individual (who is not herself included in the network).   We oriented the links in each network to point from followee to follower, corresponding to the presumed direction of information flow.  Self-loops were removed.  To ensure that fixation is possible, we eliminated individuals that could not be reached, via outgoing links, from the node with greatest eigenvector centrality.  The resulting networks varied in size from 3 to 241 nodes (Fig.~\ref{fig:twitter}).  All links were assigned equal weight, so that the molecular clock rate (as a multiple of $u$) depends only on the network topology.

We found that the mean value of $K$ among these ego networks is $0.557u$, with a standard deviation of $0.222u$.  19 of the 973 networks ($2\%$) have $K>u$.  Two networks have $K=u$ exactly; each of these has $N=3$ nodes and uniform in-degree $d_i$, thus $K=u$ follows from Result 1 for these networks.  We found a weak but statistically significant negative relationship between $N$ and $K/u$ ($\mathrm{slope}\approx -0.0016, \, R \approx -0.45, \, p < 10^{-48}$).  In summary, while some Twitter ego-networks accelerate the substitution of neutral innovations, the vast majority (especially the larger ones) slow this rate.

\begin{figure}
\begin{center}
\begin{tabular}{cc}
\includegraphics[scale=0.4]{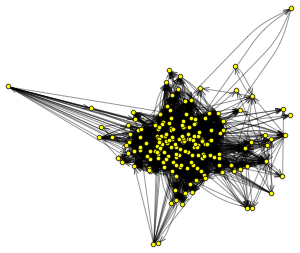} & \includegraphics[scale=0.4]{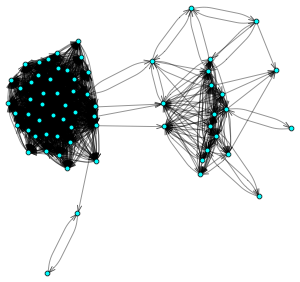}\\
Smallest $K$ & Largest $K$\\
$N=185, \; K \approx 0.088u$ & $N=66, \; K \approx 1.264u$\\
(a) & (b)\\
\includegraphics[scale=0.4]{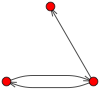} & \includegraphics[scale=0.4]{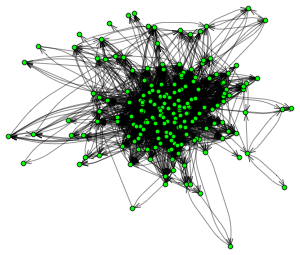}\\
$N=3, \; K =u$ & $N=199, \; K \approx 1.065u$\\
(c) & (d)\\
\includegraphics[scale=0.4]{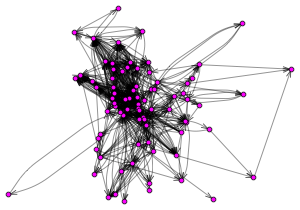} & \includegraphics[scale=0.5]{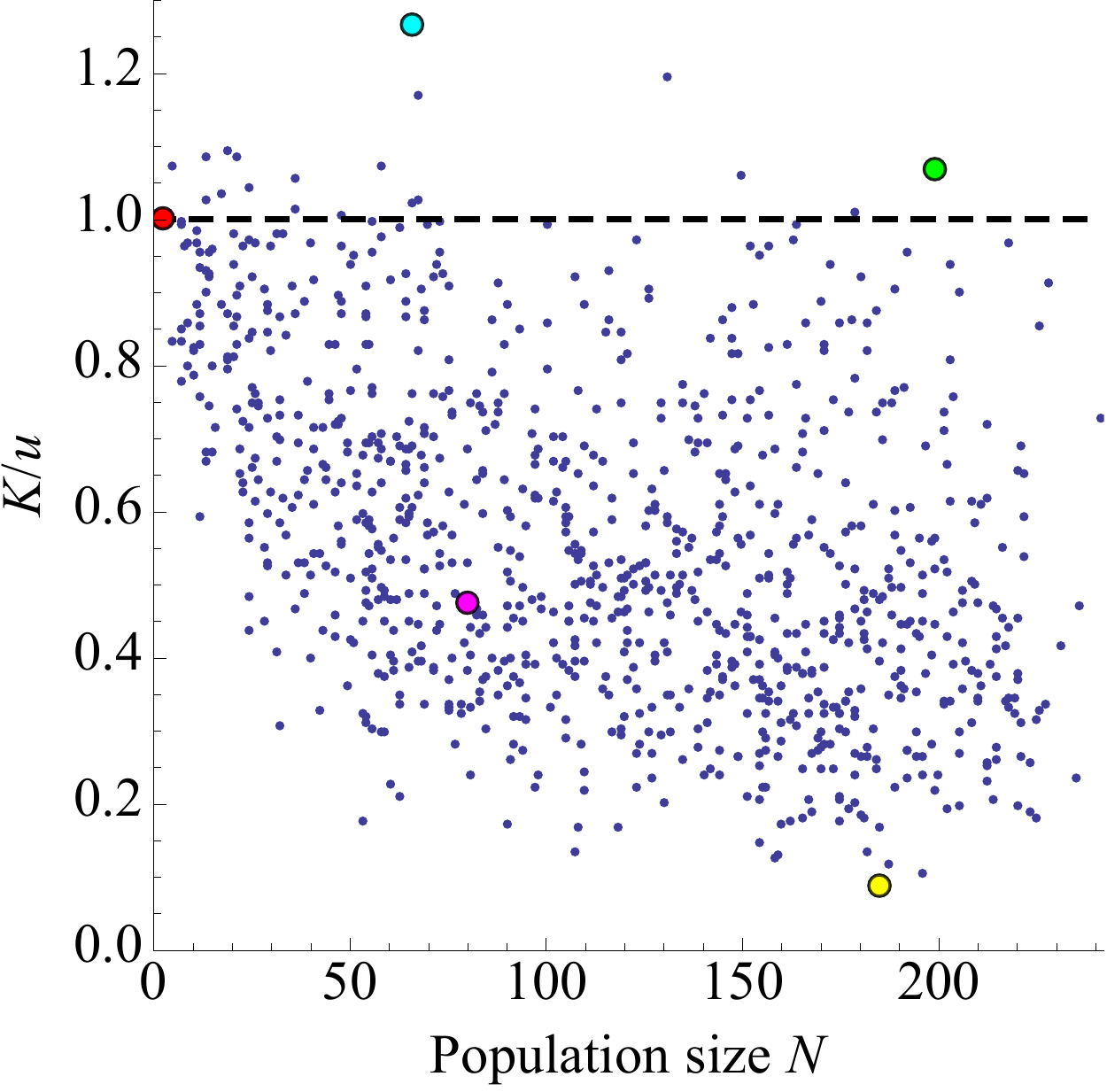}\\
$N=80, \; K \approx 0.473u$ & \\
(e) & (f)
\end{tabular}
\caption{\textbf{The rate of neutral substitution $K$ on Twitter ``ego networks".}  (a-e) Five of the 973 networks analyzed, including those with (a) the largest value of $K$, (b) the smallest value of $K$, and (c) the fewest nodes.  (f) A scatter plot of  $K/u$ versus $N$ reveals a weak negative correlation ($\mathrm{slope}\approx -0.0016, \, R \approx -0.45, \, p < 10^{-48}$). The colored dots on the scatter plot correspond to the networks shown in (a-e). The dashed line corresponds to $K/u=1$, above which network topology accelerates neutral substitution.}
\label{fig:twitter}
\end{center}
\end{figure}

\section*{Discussion}

The spatial structure of a population affects its evolution in many ways, for example by promoting cooperative behaviors \cite{NowakMay,DurrettLevin,Nakamaru,vanBaalen,Ohtsuki,nowak2010evolutionary,allen2014games,debarre2014social}, genetic variation \cite{felsenstein1976theoretical,hedrick1976genetic,slatkin1987gene,Cherry2003Diffusion}, and speciation \cite{macarthur1967theory,doebeli2003speciation,deAguiar2009global}.  Asymmetric spatial structure in particular is known to have important consequences for adaptation \cite{kawecki2002evolutionary,lenormand2002gene,Lieberman2005Graphs,antal2006evolutionary,garant2007multifarious,Houchmandzadeh2011Fixation,shakarian2012review,maciejewski2014evolutionary} and for genetic diversity \cite{Nagylaki1978eco,vuilleumier2007patch,morrissey2009maintenance}.  Our work shows that asymmetric spatial structure also affects the rate of neutral substitution.  

In light of Results 1 and 2, we see that the critical factors driving the changes in molecular clock rate are differential rates of turnover ($d_i$) and net offspring production ($b_i - d_i$) across sites.  If both $d_i$ and $b_i-d_i$ differ across sites, the molecular clock rate will in general differ from that of a well-mixed population (Result 4).  If additionally $b_i$ is constant across sites, Result 3 guarantees that neutral substitution is slowed relative to the well-mixed case.  

While our class of models includes a number of limiting assumptions (e.g., constant population size), our qualitative results do not appear to depend on these assumptions.  We expect that differential rates of turnover and net offspring production will also alter molecular clock rates in models outside of the class considered here. However, the result that molecular clock rates can come abritrarily close to $Nu$---which depends on the existence of a single ``hub" individual seeding the rest of the population as in Figure \ref{fig:noconstraint}---may require modification for sexually reproducing populations.

Conditions leading to altered molecular clock rate may occur frequently in natural populations.  Asymmetric dispersal may result, for example, from prevailing winds \cite{tackenberg2003assessment,munoz2004wind}, water currents \cite{collins2010asymmetric,Pringle2011eco}, or differences in elevation \cite{angert2009niche,okeefe2009source}.  Differences in habitat quality may lead to variance in birth and death rates across sites (nonuniform $b_i$ and $d_i$).  It is known that such asymmetries in spatial structure have important consequences for adaptation \cite{kawecki2002evolutionary,lenormand2002gene,Lieberman2005Graphs,garant2007multifarious} and for genetic diversity \cite{Nagylaki1978eco,vuilleumier2007patch,morrissey2009maintenance}; our work shows that they also have consequences for the rate of neutral genetic change.  In particular, the molecular clock is accelerated if there is greater turnover in ``upstream" subpopulations.

Many somatic cell populations have strongly asymmetric patterns of replacement, with small stem cell pools feeding much larger populations of progenitor and differentiated cells.  Our results support the idea that infrequent turnover of stem cells suppresses the accumulation of somatic mutations \cite{nowak2003linear,dingli2007stochastic,vermeulen2013defining,bozic2013unwanted,snippert2014unwanted}. The rate at which these mutations accumulate has significant implications for the onset of cancer \cite{knudson2001two,attolini2009evolutionary,vogelstein2013cancer} and the likelihood of successful cancer therapy \cite{GCbook,komarova2005drug,michor2006evolution,bozic2013evolutionary}.  It is important to note, however, that cancer can alter the structure of cell hierarchies; for example, by altering the number and replacement rate of stem cells \cite{kozar2013continuous} and/or allowing differentiated cells to revert to stem cells \cite{gupta2011stochastic}.  This restructuring may, in turn, alter the rate of genetic substitution.    

The influence of social network topology on the spread of ideas and behaviors is a question of both theoretical and practical interest \cite{may2001infection,watts2002simple,christakis2007spread,rosvall2008maps,castellano2009statistical,centola2010spread,hill2010infectious,bakshy2012role}.  The neutral substition rate $K$ on social networks describes how innovations spread when they are equally likely to be imitated as an existing convention.  Since it depends only on network topology, $K/u$ can be used as a statistic to isolate how this topology affects the rate of social change. Our finding that most Twitter ego networks have rates less than those of well-mixed populations contrasts with results from epidemiological models, which generally find that the heterogeneity of real-world social networks accelerates contagion \cite{may2001infection,moreno2002epidemic,nekovee2007theory}.  Further research is needed to determine which kinds of real-world networks accelerate neutral substitution, and how neutral substitution rates relate to other measures of diffusion and contagion on networks.  

If spatial structure remains constant over time, then the neutral substitution rate $K$ is in all cases a constant multiple of the mutation rate $u$.  In this case, absent other complicating factors such as selection, neutral mutations will accrue at a constant rate that can be inferred from genetic data.  However, if the spatial sturcture changes over time---due, for instance, to changes in climate, tumorogenesis, or social network dynamics \cite{skyrms2009dynamic,rosvall2010mapping,rand2011dynamic,wardil2014origin}---the rate of neutral substitutions may change over time as well.

In our framework, the molecular clock rate is assumed to depend only on the rate at which mutations arise and their probability of becoming fixed.  This approach assumes that the time to fixation is typically shorter than the expected waiting time $1/(Nu\rho)$ for the next successful mutation.  If this is not the case, then substitution rates are also affected by fixation times.  These fixation times are themselves affected by spatial structure \cite{Cherry2003Diffusion,Whitlock2003Fixation,PhysRevLett.101.258701,constable2014population}, leading to further ramifications for the molecular clock rate \cite{frean2013effect}.  

The starting point of our analysis is that the convention, commonly assumed in evolutionary models, that mutations arise with equal frequency at each site, is not necessarily the most natural choice. If there is a constant probability of mutation per birth, then mutations instead arise in proportion to the rate of turnover at a site.  Here we have applied this principle to study the rate of neutral substitution.  However, this principle also holds for advantageous and deleterious mutations, as well as those whose effect varies with location.  It also applies to frequency-dependent selection \cite{maciejewski2014evolutionary,tarnita2014measures}.  Re-analyzing existing models using this new mutation convention may reveal further surprises about how spatial structure affects evolution.

\section*{Methods}
\subsection*{Class of Models}

In the class of models we consider, there are $N$ sites indexed $i=1, \ldots, N$, each always occupied by a single individual.  At each time-step, a replacement event occurs, in which the occupants of some positions are replaced by the offspring of others.  A replacement event is identified by a pair $(R, \alpha)$, where $R \subset \{1, \ldots, N\}$ is the set of sites whose occupants are replaced by new offspring, and $\alpha: R \to \{1, \ldots, N\}$ is a set mapping indicating the parent of each new offspring.  (This notation was introduced in Ref.~\cite{Allen2012Success}.) A sample replacement event is illustrated in Figure \ref{fig:replacefig}.  

\begin{figure}
\begin{center}
\includegraphics{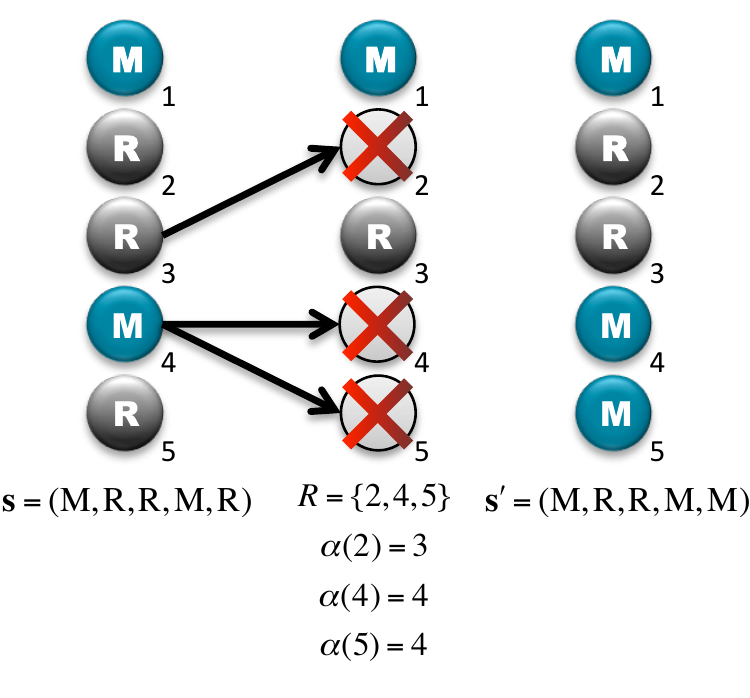}
\caption{\textbf{Illustration of a replacement event.}  In this case, the occupant of site 3 produced one offspring, which displaced the occupant of site 2.  The occupant of site 4 produced two offspring, one remaining in site 4 and displacing the parent, and the other displacing the occupant of site 5.  Thus the set of replaced positions is $R=\{2,4,5\}$ and the offspring-to-parent map $\alpha$ is given by $\alpha(2)=3, \alpha(4)=4, \alpha(5)=4$.  There is no mutation in the example illustrated here, so each offspring inherits the type of the parent.  Thus the population transitions from state $\vs=(\M,\R,\R,\M,\R)$ to $\vs'=(\M,\R,\R,\M,\M)$.}
\label{fig:replacefig}
\end{center}
\end{figure}

A model of neutral evolution is specified by a probability distribution over the set of possible replacement events.  We call this probability distribution the replacement rule of the model.  The probability of a replacement event $(R, \alpha)$ in this distribution will be denoted $p(R,\alpha)$.  Neutrality is represented by independence of the probabilities $p(R,\alpha)$ from the state of the evolutionary process.

The only assumption we place on the replacement rule is that it should be possible for at least one site $i$ to contain the eventual ancestor of the population:
\begin{assum}
\label{fixationaxiom}
There is an $i \in \{1, \ldots, N\}$, a positive integer $n$ and a finite sequence $\{ (R_k, \alpha_k) \}_{k=1}^n$ of replacement events such that
\begin{itemize}
\item $p (R_k, \alpha_k)>0$ for all $k$, and
\item For all individuals $j \in \{1, \ldots, N\}$, 
\begin{equation}
\label{eq:ancestor}
\alpha_{k_1} \circ \alpha_{k_2} \circ \cdots \circ \alpha_{k_m} (j) = i,
\end{equation}
where $k_1, \ldots < k_m$ is the maximal subsequence of $1, \ldots, n$ such that the compositions in Eq.~\eqref{eq:ancestor} are well-defined.
\end{itemize}
\end{assum}
We observe that Eq.~\eqref{eq:ancestor} traces the ancestors of $j$ backwards in time to $i$.  This assumption is equivalent to saying that there is at least one site $i$ such that mutations arising at site $i$ have nonzero fixation probability.  Assumption 1 precludes degenerate cases such as two completely separate subpopulations with no gene flow between them, in which case fixation would be impossible.

For a specific replacement event $(R, \alpha)$, the sites that are replaced by the offspring of $i$ is the given by the preimage $\alpha^{-1}(i) \subset R$, i.e., the set of indices that map to $i$ under $\alpha$.  The number of offspring of site $i$ is equal to $|\alpha^{-1}(i)|$, the cardinality (size) of this preimage.  Taking all possible replacement events into account, the birth rate (expected offspring number) of site $i$ is given by
\[
b_i = \E \big[ |\alpha^{-1}(i)| \big] = \sum_{(R, \alpha)} p(R, \alpha) \, | \alpha^{-1}(i)|.
\]
The death rate (probability of replacement) of site $i$ is equal to 
\[
d_i = \Pr[ i \in R ] = \sum_{\substack{(R, \alpha)\\ i \in R}} p (R,\alpha).
\]
The probability that the offspring of $i$ displaces the occupant of $j$ in a replacement event is
\[
e_{ij} = \Pr[ \text{$j \in R$ and $\alpha(j)=i$} ] = \sum_{\substack{(R, \alpha)\\ i \in R\\ \alpha(j)=i}} p (R,\alpha).
\]
We observe that $b_i = \sum_{j=1}^N e_{ij}$ and $d_i = \sum_{j=1}^N e_{ji}$.

\subsection*{The evolutionary Markov chain}

To study the fixation of new mutations, we consider evolution with two genetic types: mutant (M) and resident (R).  The type occupying site $i$ in a given state of the evolutionary process is denoted $s_i \in \{\M,\R\}$.  The overall state of the process can be recorded as a string $\vs = (s_1, \ldots, s_N)$ of length $N$ with alphabet $\{\M, \R\}$. 

We assume that there is no further mutation after an initial mutant appears; thus offspring faithfully inherit the type of the parent.  It follows that if the current state is $\vs = (s_1, \ldots, s_N)$ and the replacement event $(R, \alpha)$ occurs, then the new state $\vs = (s_1', \ldots, s_N')$ is given by
\[
s_i' = \begin{cases}
s_i & i \notin R\\
s_{\alpha(i)} & i \in R. \end{cases}
\]

The above assumptions describe a Markov chain on the set of strings of length $N$ with alphabet $\{\M, \R\}$. We call this the evolutionary Markov chain.  It is straightforward to show that, from any initial state, this Markov chain will eventually converge upon one of two absorbing states: $(\M, \ldots, \M)$ or $(\R, \ldots, \R)$ \cite{Allen2012Success}.  In the former case, we say that the mutant type has gone to fixation; in the latter case we say that the mutant type has disappeared.

\subsection*{The ancestral Markov chain}

It is useful to consider a variation on the evolutionary Markov chain called the \emph{ancestral Markov chain}, denoted $\cA$. Instead of two types, the ancestral Markov chain has $N$ types, labelled $1, \ldots, N$, which correspond to the $N$ members of a ``founding generation" of the population.  Evolution proceeds according to the given replacement rule, as described in the Material and Methods. The states of the ancestral Markov chain are strings of length $N$ with alphabet $\{1, \ldots, N\}$.

The ancestral Markov chain has a canonical initial state $\va_0=(1, \ldots, N)$, in which the type of each individual  corresponds to its location.  This initial state identifies the locations of each founding $(t=0)$ member of the population.  The ancestral Markov chain with initial state $\va_0$---in our notation, $(\cA, \va_0)$---has the useful feature that at any time $t \geq 0$, the state $\va(t) =(a_1(t), \ldots, a_N(t))$ indicates the site occupied by each individual's founding  ancestor.  In other words, if $a_j(t)=i$, then the current occupant of site $j$ is descended from the founder that occupied site $i$.

To relate the evolutionary Markov chains $\cA$ and $\cM$, consider any set mapping $\gamma:  \{1, \ldots, N\} \to \{\M, \R\}$.  We think of $\gamma(i)$ as giving the genetic type (M or R) of each member $i=1, \ldots, N$ of the founding generation.  The mapping $\gamma$ induces a mapping $\tilde{\gamma}$ from states of $\cA$ to states of $\cM$, defined by $\tilde{\gamma}(a_1, \ldots, a_N)=(\gamma(a_1), \ldots, \gamma(a_N))$.  For any state $\va(t)$ of  $(\cA,\va_0)$, the string $\tilde{\gamma}(\va(t)) \in \{1, \ldots, N\}^N$ indicates the genetic type of each individual's ancestor in the founding generation.  Since genetic types are inherited faithfully, it follows that $\tilde{\gamma}(\va(t))$ gives the current genetic type of each individual.  Thus if $\cM$ and $\cA$ follow the same replacement rule, we have that for any such mapping $\gamma$ and any string $\vs \in \{\M,\R\}^N$,
\begin{equation}
\label{eq:MArelation}
\Prob_{(\cA, \va_0)} \, [ \tilde{\gamma} (\va(t)) = \vs ] \quad = \quad 
\Prob_{(\cM, \tilde{\gamma}(\va_0))} \, [ \vs(t) = \vs ].
\end{equation}

\subsection*{Basic results on fixation probabilities}

We define the \emph{fixation probability from site $i$}, $\rho_i$, to be the probability that, from an initial state a mutant in site $i$ and residents in all other sites, the mutant type goes to fixation:
\begin{equation}
\label{eq:rhoidef}
\rho_i  = \lim_{t \to \infty} \Prob_{(\cM, \vm_i)} \, [ \vs(t) = (\M, \ldots, \M) ] 
\end{equation}
Above, $\vm_i$ denotes the initial state consisting of an M in position $i$ and R's elsewhere.  The ordered pair $(\cM, \vm_i)$ refers to the Markov chain $\cM$ with initial state $\vm_i$. $\Prob_{(\cM, \vm_i)}[\vs(t)=\vs]$ denotes the probability that the state of $(\cM, \vm_i)$ is $\vs$ at time $t \geq 0$.  

We can use the relationship between the ancestral Markov chain $\cA$ and the evolutionary Markov chain $\cM$ to obtain an alternate expression for the site-specific fixation probability $\rho_i$:
\begin{theorem}
\label{lem:rhoaltdef}
$\rho_i = \lim_{t \to \infty} \Prob_{(\cA, \va_0)} \, [ \va(t) = (i, \ldots, i) ]$.
\end{theorem}

In words, the site-specific fixation probability $\rho_i$ equals the probability that founding individual $i$ becomes the eventual ancestor of the whole population.  

\begin{proof}
For any $i=1, \ldots, N$, define the set mapping $\gamma_i:\{1, \ldots, N\} \to \{\M,\R\}$ by
\[
\gamma_i(j) = \begin{cases} \M & \text{if $j=i$}\\
\R & \text{otherwise.}
\end{cases}
\]
(Intuitively, this mapping describes the case that individual $i$ of the founding generation is a mutant, and all others in the founding generation are residents.)  Note that $\tilde{\gamma}_i(\va_0) = \vm_i$.
Combining Eq.~\eqref{eq:MArelation} with the definition of $\rho_i$, we obtain
\begin{align}
\nonumber
\rho_i & = \lim_{t \to \infty} \Prob_{(\cM, \vm_i)} \, [ \vs(t) = (\M, \ldots, \M) ] \\
\nonumber
& = \lim_{t \to \infty} \Prob_{(\cA, \va_0)} \, [ \tilde{\gamma_i}(\va(t)) = (\M, \ldots, \M) ] \\
\label{eq:rhoialtdef}
& = \lim_{t \to \infty} \Prob_{(\cA, \va_0)} \, [ \va(t) = (i, \ldots, i) ]. \qedhere
\end{align}
\end{proof}

More generally, we can consider the probability of fixation from an arbitrary set of sites.  For any set $S \subset \{1, \ldots, N\}$, we let $\rho_S$ denote the probability that the mutant type becomes fixed, given the initial state $\vm_S$ with mutants occupying the sites specified by $S$ and residents occupying all other sites:
\begin{equation*}
\rho_S  = \lim_{t \to \infty} \Prob_{(\cM, \vm_S)} \, [ \vs(t) = (\M, \ldots, \M) ] 
\end{equation*}

Site-specific fixation probabilities are additive in the following sense:

\begin{theorem} \label{thm:rhoset}
  For any set $S \subset \{ 1, \ldots N \}$ of sites, 
\begin{equation}
\label{eq:rhoset}
\rho_S = \sum_{i \in S} \rho_i.
\end{equation}
In particular, 
\begin{equation}
\label{eq:rhosum}
\sum_{i=1}^N \rho_i = 1.
\end{equation}
\end{theorem}

This result has previously been obtained for a number of specific evolutioanry processes on graphs \cite{broom2010two,shakarian2013novel,maciejewski2014reproductive}.  Intuitively, the probability of fixation from the initial state described by $S$ equals the probability that one of the individuals in a site $i \in S$ becomes the eventual ancestor of the population.  Since this cannot be true of more than one site in $S$, the overall probability $\rho_S$ is obtained by summing over all $i \in S$ the individual probabilities $\rho_i$ that site $i$ contains the eventual ancestor.  

\begin{proof}
Suppose that we are given a set $S \subset \{1, \ldots, N\}$ of sites initially occupied by mutants.  This situation is described by the set mapping $\gamma_S:\{1, \ldots, N\} \to \{\M,\R\}$, given by
\[
\gamma_S(i) = \begin{cases} \M & \text{if $i \in S$}\\
\R & \text{otherwise.}
\end{cases}
\]
Invoking the relationship between $\cA$ and $\cM$---in particular, Eq.~\eqref{eq:MArelation} and Theorem \ref{lem:rhoaltdef}---we have
\begin{align*}
\rho_S & = \lim_{t \to \infty} \Prob_{(\cM, \tilde{\gamma_S}(\va_0))} \, [ \vs(t) = (\M, \ldots, \M) ] \\
& = \lim_{t \to \infty} \Prob_{(\cA, \va_0)} \, [ \tilde{\gamma_S}(\va(t)) = (\M, \ldots, \M) ] \\
& = \lim_{t \to \infty} \Prob_{(\cA, \va_0)} \, [ \text{$\va(t) = (i, \ldots, i)$ for some $i \in S$}]\\
& = \sum_{i \in S} \lim_{t \to \infty} \Prob_{(\cA, \va_0)} \, [ \va(t) = (i, \ldots, i)]\\
& = \sum_{i \in S} \rho_i.
\end{align*}
This proves Eq.~\eqref{eq:rhoset}.  Eq.~\eqref{eq:rhosum} now follows from letting $S=\{1, \ldots, N\}$, and noting that, in this case,
\[
\rho_S = \lim_{t \to \infty} \Prob_{(\cM, (\M, \ldots, \M))} \, [ \vs(t) = (\M, \ldots, \M) ] = 1. \qedhere
\]
\end{proof}

We now derive Eq.~\eqref{eq:rhorecur}, which allows the fixation probabilities $\rho_i$ to be calculated from the replacement rates $e_{ij}$.  

\begin{theorem}
For each $i = 1, \ldots, N$,
\[
d_i \rho_i =\sum_{j=1}^N e_{ij} \rho_j.
\]
\end{theorem}

\begin{proof}
Considering the change that can occur over a single time-step, we have the following recurrence relation:
\[
\rho_i = (1-d_i) \rho_i + \sum_{(R,\alpha)} p(R, \alpha) \, \rho_{\alpha^{-1}(i)}.
\]
The first term above represents the case that the occupant of $i$ survives the current time-step and becomes the eventual ancestor of the population, while the second term represents the case that one of $i$'s offspring from the current time-step is the eventual ancestor.  Subtracting $(1-d_i) \rho_i$ from both sides and applying Theorem \ref{thm:rhoset} with $S=\rho_{\alpha^{-1}(i)}$ yields
\[
d_i \rho_i =  \sum_{(R, \alpha)} p(R, \alpha) \left( \sum_{j \in \alpha^{-1}(i)}  \rho_j \right).
\]
Now interchanging the summations on the right-hand side yields
\[
d_i \rho_i = \sum_{j=1}^N \rho_j 
\left ( \sum_{\substack{(R,\alpha)\\j \in R \\ \alpha(j)=i}} p(R, \alpha) \right ).
\]
By definition, 
\[
\sum_{\substack{(R,\alpha)\\j \in R \\ \alpha(j)=i}} p(R, \alpha) = e_{ij}. 
\]
This completes the proof.
\end{proof}

\subsection*{Proof of Result 2}

\begin{theorem}[Result 2]
The fixation probabilities from each site are equal to $1/N$ ($\rho_i=1/N$ for all $i=1, \ldots, N$) if and only if each site has birth rate equal to death rate ($b_i = d_i$ for all $i=1, \ldots, N$). 
\end{theorem}

\begin{proof}
Assume first that  the fixation probabilities from each site are all equal to $1/N$.  Substituting $\rho_i=1/N$ for all $i$ into Eq.~\eqref{eq:rhorecur} yields $d_i = b_i$.  This proves the ``only if" direction.

Next assume that the birth rate is equal to the death rate at each site  ($b_i = d_i$ for all $i=1, \ldots, N$).   Eq.~\eqref{eq:rhorecur} can then be rewritten as 
 \begin{equation*}
\sum_{j=1}^N e_{ij}\rho_{i} =\sum_{j=1}^N e_{ij}\rho_{j}.
\end{equation*}
Clearly, $\rho_i=1/N$ for all $i$ satisfies the above equation for all $i$, and also satisfies $\sum_{i=1}^N \rho_i = 1$.   Assumption 1 guarantees that the solution to these equations is unique.  Therefore $b_i=d_i$ for all $i$ implies $\rho_i=1/N$ for all $i$, proving the ``if" direction.
\end{proof}

\subsection*{Proof of Result 3}

\begin{theorem}[Result 3]
If birth rates are constant over sites, $b_i=1/N$ for all $i=1,\ldots,N$,
then $\rho \leq 1/N$ (and consequently $K \leq u$) with equality if and only if the death rates are also constant over sites.  
\end{theorem}

\begin{proof} We separate our proof into five steps.  First, we show that there is no loss of generality in assuming that $B=1$.  Second, we use the method of Lagrange multipliers to find the critical points of the function $\rho = \sum_{i=1}^N d_i \rho_i$ with respect to the variables $\{e_{ij}\}_{i,j=1}^N$ and $\{\rho_i\}_{i=1}^N$ and the constraints
\begin{equation}
\begin{cases}
\displaystyle \sum_{j=1}^N e_{ij}  = \frac{1}{N} & i \in \{1,\ldots, N\}, \\
\displaystyle \left(\sum_{j=1}^N e_{ji} \right) \rho_i  = \sum_{j=1}^N e_{ij} \rho_j & i \in \{1,\ldots, N\}, \\
\displaystyle \sum_{i=1}^N \rho_i  = 1.
\end{cases}
\end{equation}
We obtain that the critical points are precisely those for which $d_i = 1/N$ for all $i$.  In the next three steps, we use partial derivatives of Eqs.~\eqref{eq:rhorecur} and \eqref{eq:rhodef} to form a second-order Taylor expansion of $\rho$ in the variables $\{e_{ij}\}_{ i \neq j}$ around the critical points.  We show that the second-order part of this expansion is $\leq 0$ with equality only at the critical points, completing the proof.

\subsubsection*{Step 1: Normalize the expected number of offspring per time-step}

We first show that there is no loss of generality in assuming that the total expected number of offspring per time-step, $B \equiv \sum_{i,j} e_{ij}$, is one.  We will demonstrate this by showing that, if a maximum of $\rho$ is achieved by a combination of replacement rates $\{e_{ij}^*\}_{i,j}$, this maximum is also achieved using the normalized rates $\{e_{ij}'\}_{i,j}$ with $e_{ij}'= e_{ij}^*/B^*$, $B^* \equiv \sum_{i,j} e_{ij}$.  

First, we see that the equations
\[
d_i^* \rho_i = \sum_{j=1}^N e_{ij}^* \rho_j \qquad \text{and} 
\qquad d_i' \rho_i = \sum_{j=1}^N e_{ij}' \rho_j
\]
are equivalent, differing only by a factor of $B$.  (Above, we have introduced the notation $d_i^* = \sum_{j=1}^N e_{ij}^*$ and $d_i' = \sum_{j=1}^N e_{ij}'$.)  Thus the node-specific fixation probabilities $\{ \rho_i \}_{i=1}^N$ resulting from $\{e_{ij}'\}_{i,j}$ are the same as those resulting from $\{e_{ij}^*\}_{i,j}$.

Turning now to the overall fixation probabilities $\rho^*$ and $\rho'$ corresponding to $\{e_{ij}^*\}_{i,j}$ and $\{e_{ij}'\}_{i,j}$ respectively, we see that 
\[
\rho^* = \frac{1}{B^*} \sum_{i=1}^N d_i^* \rho_i =  \sum_{i=1}^N \frac{d_i^*}{B^*} \rho_i 
= \sum_{i=1}^N d_i' \rho_i = \rho'.
\]
Thus the same maximum is achieved by the normalized replacement rates $\{e_{ij}'\}_{i,j}$, which satisfy $\sum_{i,j} e_{ij}'=1$.  We conclude that there is no loss of generality in assuming $B=1$.  Together with uniform birth rates, this implies the constraint
\begin{equation}
\label{eq:birth_const}
\sum_{j=1}^N e_{ij} = \frac{1}{N}.
\end{equation}

\subsubsection*{Step 2: Determine critical points}
We use the method of Lagrange multipliers to find critical points of the function $\rho = \sum_{i=1}^N d_i \rho_i$ with respect to the variables $\{e_{ij}\}_{i,j=1}^N$ and $\{\rho_i\}_{i=1}^N$.  Eqs.~\eqref{eq:birth_const}, \eqref{eq:rhorecur}, and \eqref{eq:rhosum} give the respective constraint equations
\begin{align*}
g_{i}&\equiv \sum_{j=1}^{N}e_{ij}-\frac{1}{N}=0,\ \text{for}\, i \in \{1,\ldots, N\}\\
h_{i}&\equiv \rho_{i}\left(\sum_{j=1}^{N}e_{ji}\right)-\sum_{i=1}^{N}e_{ij}\rho_{j}=0,\ \text{for}\, i \in \{1,\ldots, N\}\\
q&\equiv \sum_{i=1}^{N}\rho_{i}-1=0
\end{align*}
Critical points with respect to the above constraints are solutions to
\begin{equation*}
\label{eq:lagrange}
\nabla{\rho}=\sum_{i=1}^{N}\lambda_{i}\nabla{g_{i}} + \sum_{i=1}^{N}\mu_{i}\nabla{h_{i}} + \sigma\nabla{q},
\end{equation*}
where the $\lambda_{i}$'s, $\mu_{i}$'s, and $\sigma$ are the Lagrange multipliers.  The individual components of the gradient are obtained by taking the partial derivative with respect to $e_{kl}$, yielding  
\begin{equation}
\rho_{l}=\lambda_{k}+\rho_{l}\mu_{l}-\rho_{l}\mu_{k},
\label{eq:lagrange_partialA}
\end{equation}
for $k,l\in\{1,\ldots,N\}$, and the partial derivative with respect to $\rho_k$, yielding
\begin{equation}
\sum_{j=1}^{N}e_{jk} = \mu_{k}\sum_{j=1}^{N}e_{jk}-\sum_{j=1}^{N}\mu_{j}e_{jk}+\sigma,
\label{eq:lagrange_partialB}
\end{equation}
for $k=1, \ldots, N$.  Solving Eq.~\eqref{eq:lagrange_partialA} for $\rho_l$ gives
\begin{equation}
\rho_{l}=\frac{\lambda_{k}}{1-\mu_{l}+\mu_{k}}.\label{eq:rhol_with_lambda}
\end{equation}
Note that in the case $l=k$, Eq.~\eqref{eq:rhol_with_lambda} becomes $\rho_k=\lambda_k$.   Therefore, we replace $\lambda_k$ with $\rho_k$   in Eq.~\eqref{eq:rhol_with_lambda}:
\begin{equation}
\rho_{l}=\frac{\rho_{k}}{1-\mu_{l}+\mu_{k}}\label{eq:rhol}.
\end{equation}
Interchanging $k$ and $l$, we obtain
\begin{equation}
\rho_{k}=\frac{\rho_{l}}{1-\mu_{k}+\mu_{l}}.\label{eq:rhok}
\end{equation}

Combining Eqs.~\eqref{eq:rhol} and \eqref{eq:rhok} yields $\mu_l = \mu_k$ and $\rho_l=\rho_k$ for all $k,l\in\{1,\ldots,N\}$.  That is, all node-specific probabilities $\rho_i$ are equal.  It follows from   Eq.~\eqref{eq:rhosum} that $\rho_i=1/N$ for all  $i=1,\ldots,N$.
Furthermore, given that $\mu_l = \mu_k$, Eq.~\eqref{eq:lagrange_partialB} yields $d_k=\sigma$ for all $k=1,\ldots,N$.  That is, all death rates $d_i$ are equal. Therefore $d_i = 1/N$ for $i = 1,\ldots,N$  because we have assumed (without loss of generality) that $\sum_{i=1}^N d_i = 1$.

In summary, we have shown that the overall fixation probability $\rho$ has a critical point whenever all node-specific fixation probabilities and death rates are constant ($d_i = \rho_i = 1/N$ for all $i=1,\ldots,N$).  Consequently, $\rho= 1/N$ at all critical points. It still remains to prove that this critical value of $\rho$ is a maximum.

\subsubsection*{Step 3: Taylor series expansion}
To prove that the overall fixation probability has an absolute maximum of $1/N$, we pick a critical point $\{e_{ij}^*\}_{ i, j}$.  Then, viewing $\rho$ as a function of the independent variables $\{e_{ij}\}_{\substack{ i, j\\ i \neq j}}$, we form a second-order Taylor expansion of $\rho$ around the critical point.  (Since we are operating under the constraint $\sum_{j=1}^N e_{ij} = 1/N$ for all $i$, this set of variables suffices to determine the value of $\rho$.)  We will show that the second-order term of this Taylor expansion is always less than or equal to zero.  

The second order multivariate Taylor series expansion for $\rho$ about the critical point $\{e_{ij}^*\}_{\substack{ i, j\\ i \neq j}}$  can be written as
\begin{equation}
\label{eq:Taylor}
\rho = \frac{1}{N}+\sum_{\substack{k,l\\k\neq l}}\frac{\partial{\rho}}{\partial{e_{kl}}}\Delta{e_{kl}}
+\frac{1}{2}\sum_{\substack{k,l,m,n\\m\neq n\\k\neq l}}\frac{\partial^{2}{\rho}}{\partial{e_{kl}\partial{e_{mn}}}}\Delta{e_{kl}}\Delta{e_{mn}}
+\mathcal{O}\left(|\Delta\vec{e}|^{3}\right),
\end{equation}
where all derivatives are taken at $\{e_{ij}^*\}_{\substack{ i, j\\ i \neq j}}$ and $\Delta e_{kl} = e_{kl} - e_{kl}^*$.  More simply, we write this expansion as  
\begin{equation*}
\rho=\frac{1}{N} + \rho^{(1)} + \rho^{(2)} + \mathcal{O}\left(|\Delta\vec{e}|^{3}\right),
\end{equation*}
with $\rho^{(1)}$ and $\rho^{(2)}$ representing the first- and second-order terms, respectively, in Eq.~\eqref{eq:Taylor}.
The first-order term $\rho^{(1)}$ is zero since $\{e_{ij}^*\}_{\substack{ i, j\\ i \neq j}}$ is a critical point.   Our goal will be to show that the second-order term, $\rho^{(2)}$ is always negative or zero.

We can find an alternative expression for $\rho^{(2)}$ using the definition of $\rho$ in Eq.~\eqref{eq:rhodef} as follows.
We introduce the notation $\Delta d_i = d_i - d_i^*= d_i -1/N$.  With this notation, Eq.~\eqref{eq:rhodef} becomes
\begin{align}
\nonumber
\rho&=\frac{1}{N}\sum_{i=1}^N  \rho_i + \sum_{i=1}^N \Delta d_i \rho_i\\
\label{eq:rhodeltad}
&= \frac{1}{N}+\sum_{i=1}^N \Delta d_i \rho_i.
\end{align} 
We substitute the first-order Taylor series expansion for  $\rho_i$,
\begin{equation*}
\rho_i=\frac{1}{N} + \rho_i^{(1)} + \mathcal{O}\left(|\Delta\vec{e}|^{2}\right),
\end{equation*}
into Eq.~\eqref{eq:rhodeltad}, noting that $\Delta d_i = \mathcal{O}\left(|\Delta\vec{e}|\right)$ and $\sum_{i=1}^N \Delta d_i=0$. This yields
\begin{align*}
\rho &= \frac{1}{N}+\sum_{i=1}^N \Delta d_i\left(\frac{1}{N} + \rho_i^{(1)} + \mathcal{O}\left(|\Delta\vec{e}|^{2}\right)\right)\\
&= \frac{1}{N} + \sum_{i=1}^N \Delta d_i\rho_i^{(1)}  + \mathcal{O}\left(|\Delta\vec{e}|^{3}\right).
\end{align*}
Thus, the second-order term of the overall fixation probility can be written as
\begin{equation}
\label{eq:rho2}
\rho^{(2)}=\sum_{i=1}^N \Delta d_i \rho_{i}^{(1)}.
\end{equation}

\subsubsection*{Step 4: Determine first-order term of the site-specific fixation probability}
We now investigate the properties of $\rho_i^{(1)}$.  We begin by rewriting Eq.~\eqref{eq:rhorecur} as 
\begin{equation*}
d_{i}\rho_{i}=e_{ii}\rho_{i}+\sum_{j \neq i}e_{ij}\rho_{j}.
\end{equation*}

Replacing $e_{ii}$ with $\frac{1}{N} - \sum_{j \neq i} e_{ij}$ and simplifying, we obtain
\begin{equation*}
\Delta d_i \rho_i =\sum_{j \neq i}e_{ij}\left(\rho_{j}-\rho_{i}\right).
\end{equation*}

Next we take the partial derivative of both sides with respect to $e_{kl}$, where $k\ne l$,
\begin{equation*}
\rho_{i}\frac{\partial (\Delta d_i)}{\partial{e_{kl}}}+\Delta d_i\frac{\partial{\rho_i}}{\partial{e_{kl}}}=\sum_{\substack{i,j\\j \neq i}}\frac{\partial{e_{ij}}}{\partial{e_{kl}}}\left(\rho_{j}-\rho_{i}\right)+\sum_{\substack{i,j\\j \neq i}}e_{ij}\left(\frac{\partial{\rho_{j}}}{\partial{e_{kl}}}-\frac{\partial{\rho_{i}}}{\partial{e_{kl}}}\right),
\end{equation*}
and evaluate at the critical point ($\rho_i^* = \frac{1}{N}$ and $\Delta d_i=0$):
\begin{equation*}
\frac{1}{N}\frac{\partial (\Delta d_i)}{\partial{e_{kl}}}=\sum_{j \neq i}e_{ij}^* \left(\frac{\partial{\rho_{j}}}{\partial{e_{kl}}}-\frac{\partial{\rho_{i}}}{\partial{e_{kl}}}\right).
\end{equation*}
Multiplying both sides by $\Delta e_{kl} $ and then summing over all $k, l\in\{1,\ldots,N\}$ for $k\ne l$ yields
\begin{equation}
\label{eq:rhorecur_sum}
\frac{1}{N}\sum_{\substack{k,l\\k \neq l}}\frac{\partial (\Delta d_i)}{\partial{e_{kl}}}\Delta{e_{kl}}=\sum_{\substack{j,k,l\\k \neq l\\j \neq i}}e_{ij}^*\frac{\partial{\rho_{j}}}{\partial{e_{kl}}}\Delta{e_{kl}}-\sum_{\substack{j,k,l\\k \neq l\\j \neq i}}e_{ij}^*\frac{\partial{\rho_{i}}}{\partial{e_{kl}}}\Delta{e_{kl}}.
\end{equation} 

We observe that $\rho_{i}^{(1)}=\sum_{\substack{k,l\\k\neq l}}\frac{\partial{\rho_{i}}}{\partial{e_{kl}}}\Delta{e_{kl}}$  is the first order term of the Taylor series expansion of $\rho_i$ about the critical point, and  $\sum_{\substack{k,l\\k \neq l}}\frac{\partial (\Delta d_i)}{\partial{e_{kl}}}\Delta{e_{kl}}= \Delta d_i$, since $\Delta d_i$ is a linear function of the $e_{kl}$.  Thus Eq.~\eqref{eq:rhorecur_sum} becomes
\begin{equation*}
\frac{1}{N}\Delta d_i =\sum_{j \neq i}e_{ij}^* \rho_{j}^{(1)}-\sum_{j \neq i}e_{ij}^* \rho_{i}^{(1)}.
\end{equation*}
The restriction $j \neq i$ in the sums on the right-hand side can be removed, since the entire right-hand side is zero when $j=i$.  Rearranging further, we have
\begin{equation*}
N\rho_{i}^{(1)} \sum_{j=1}^N e^*_{ij}=-\Delta d_i+N\sum_{j=1}^N e^*_{ij}\rho_{j}^{(1)}.
\end{equation*}
Since the birth rate is uniform over sites $(b_i=\sum_{j=1}^N e_{ij}^* = \frac{1}{N}$ for all $i$), the above equation simplifies to
\begin{equation}
\label{eq:rhoi_first}
\rho_{i}^{(1)}=-\Delta d_i+N\sum_{j=1}^N e_{ij}^* \rho_{j}^{(1)}.
\end{equation}
This equation provides a recurrence relation for the first-order term of the Taylor expansion of $\rho_i$ about the critical point.  

\subsubsection*{Step 5: Determine second-order term of the overall fixation probability}
Finally, we combine Eq.~\eqref{eq:rhoi_first} with Eq.~\eqref{eq:rho2} to show that $\rho^{(2)} \leq 0$.  The process is as follows.  We first multiply both sides of Eq.~\eqref{eq:rhoi_first} by $\rho_i^{(1)}$ and sum over $i=1, \ldots, N$:
\begin{equation*}
\sum_{i=1}^N \left(\rho_{i}^{(1)}\right)^{2}=-\sum_{i=1}^N \Delta d_i\rho_{i}^{(1)}
+ N\sum_{i,j}e_{ij}^* \rho_{i}^{(1)}\rho_{j}^{(1)}.
\end{equation*}
By Eq.~\eqref{eq:rho2} we can substitute $\rho^{(2)}$ for $\sum_{i=1}^N \Delta d_i \rho_{i}^{(1)}$.  Then, solving for $\rho^{(2)}$, we obtain
\begin{eqnarray}
\rho^{(2)}=-\sum_{i=1}^N \left(\rho_{i}^{(1)}\right)^{2}+N\sum_{i,j}e_{ij}^* \rho_{i}^{(1)}\rho_{j}^{(1)}.\label{eq:rho_secondA}
\end{eqnarray}
We now make use of the fact that the product $\rho_{i}^{(1)}\rho_{j}^{(1)}$ can be written as a difference of squares:
\begin{equation*}
\rho_{i}^{(1)}\rho_{j}^{(1)}=
-\frac{1}{2}\left[\left(\rho_{i}^{(1)}-\rho_{j}^{(1)}\right)^{2}
-\left(\rho_{i}^{(1)}\right)^{2}-\left(\rho_{j}^{(1)}\right)^{2}\right].
\label{eq:diff_squares}
\end{equation*}
Substituting this identity into Eq.~\eqref{eq:rho_secondA} yields
\begin{align*}
\rho^{(2)}& = -\sum_{i=1}^N \left(\rho_{i}^{(1)}\right)^{2}-\frac{N}{2}\sum_{i,j}e_{ij}^*\left(\rho_{i}^{(1)}-\rho_{j}^{(1)}\right)^{2}+\frac{N}{2}\sum_{i,j}e_{ij}^*\left(\rho_{i}^{(1)}\right)^{2}\\
& \qquad +\frac{N}{2}\sum_{i,j}e_{ij}^*\left(\rho_{j}^{(1)}\right)^{2}\\
&= -\sum_{i=1}^N \left(\rho_{i}^{(1)}\right)^{2}-\frac{N}{2}\sum_{i,j}e_{ij}^*\left(\rho_{i}^{(1)}-\rho_{j}^{(1)}\right)^{2}+\frac{N}{2}\sum_{i=1}^N \left(\rho_{i}^{(1)}\right)^{2} \sum_{j=1}^N e_{ij}^*\\
& \qquad +\frac{N}{2}\sum_{j=1}^N \left(\rho_{j}^{(1)}\right)^{2} \sum_{i=1}^N e_{ij}^*.
\end{align*}
Uniform birth rates implies that $\sum_{j=1}^N e^*_{ij} = 1/N$ for each $i$.  Furthermore, since $\{ e_{ij}^*\}_{i,j}$ is a critical point, death rates are uniform as well; thus  $\sum_{i=1}^N  e^*_{ij} = 1/N$ for each $j$.  Making these substitutions yields
\begin{align*}
\rho^{(2)}& = -\sum_{i=1}^N \left(\rho_{i}^{(1)}\right)^{2}
-\frac{N}{2}\sum_{i,j} e_{ij}^* \left(\rho_{i}^{(1)}-\rho_{j}^{(1)}\right)^{2}
+\frac{1}{2}\sum_{i=1}^N \left(\rho_{i}^{(1)}\right)^{2}
+\frac{1}{2}\sum_{j=1}^N \left(\rho_{j}^{(1)}\right)^{2}\\
& = -\frac{N}{2} \sum_{i,j} e_{ij}^* \left(\rho_{i}^{(1)}-\rho_{j}^{(1)}\right)^{2}.
\end{align*}
We conclude that $\rho^{(2)}$ is always less than or equal to $0$.  

Overall, we have shown that the critical points of $\rho$, subject to the constraint $b_i = 1/N$ for all $i$, are precisely those for which $d_i = 1/N$ for all $i$, and that $\rho$ achieves its maximum value of $1/N$ at these points.
\end{proof}

\subsection*{Result 4}

We now turn to the full range of possible values for $K$ with no constraints on birth and death rates.  Consider the example spatial structure illustrated in Figure \ref{fig:noconstraint}, which consists of a hub with outgoing edges to $n$ leaves, so that the population size is $N=n+1$.   There is an edge of weight $a$, $0 \leq a < 1/(N-1)$, from the hub to each leaf. The hub also has a self-loop of weight $1-(N-1)a$, so that $B=1$ births are expected per time-step. The death rates are $d_\mathrm{H} = 1-(N-1)a$ for the hub and $d_\mathrm{L} =a$ for each leaf. Solving  Eq.~\eqref{eq:rhorecur} we obtain the site-specific fixation probabilities $\rho_\mathrm{L} = 0$ for each leaf and $\rho_\mathrm{H}=1$ for the hub. By Eq.~\eqref{eq:rhodef}, the overall fixation probability is equal to the rate of turnover at the hub:
 \begin{equation*}
\rho = d_\mathrm{H} = 1-(N-1)a.
\end{equation*}

Any value $0 \leq \rho < 1$ can be obtained by an appropriate choice of $a$, with $0 < a \leq 1/(N-1)$, specifically, $a=(1-\rho)/(N-1)$.  This proves:
 
\begin{theorem}[Result 4]
For arbitrary spatial population structure (no constraints on $e_{ij}$) the fixation probability can take any value $0 \leq \rho < 1$, and consequently, the molecular clock can take any rate $0 \leq K < Nu$.
\end{theorem}

We observe that if $a=1/N$ then the death rate is $1/N$ at each site, and therefore $\rho= 1/N$ by Result 1.  If $a=1/(N-1)$ then there is no turnover at the hub and thus $\rho=0$.   At the other extreme, as $a$ approach zero, $\rho$  comes arbitrarily close to unity.

\subsection*{Exact expressions for $\rho$ with $N \leq 3$}

To complement the above arguments, which apply to general population size $N$, we here derive exact expressions for the overall fixation probability $\rho$ in the cases $N=2$ and $N=3$.   Our results for node-specific fixation probabilities coincide with those found previous for a different model of evolution in a subdivided population \cite{Lundy1998eco}.

\subsubsection*{Population size $N=2$}

We solve Eq.~\eqref{eq:rhorecur} to obtain expressions for $\rho_1$ and $\rho_2$:
\begin{align}
\rho_{1}&=\frac{e_{12}}{e_{12}+e_{21}},\label{eq:rho1_2node}\\
\rho_{2}&=\frac{e_{21}}{e_{12}+e_{21}}.\label{eq:rho2_2node}
\end{align}

We first consider uniform birth rates ($b_i=\sum_{j=1}^2 e_{ij} = \frac{1}{2}$ for $i=1,2$).  Substituting Eqs.~\eqref{eq:rho1_2node} and \eqref{eq:rho2_2node} into Eq.~\eqref{eq:rhodef} ($B=1$ in this case), gives an expression for the overall fixation probability:
\begin{equation*}
\rho=\frac{1}{2}-\frac{(e_{12}-e_{21})^2}{e_{12}+e_{21}}.
\end{equation*}
Thus, in accordance with Result 3, $\rho$ is less than or equal to $\frac{1}{2}$ and consequently $K \leq u$.  Equality occurs if and only if $e_{12}=e_{21}$, which corresponds to equal node-specific fixation probabilities ($\rho_1=\rho_2=1/N$ from Eqs.~\eqref{eq:rho1_2node} and \eqref{eq:rho2_2node}) and equal death rates ($d_1=d_2=1/N$).

When the condition of uniform birth rates is relaxed, we find, by substituting Eqs.~\eqref{eq:rho1_2node} and \eqref{eq:rho2_2node} into Eq.~\eqref{eq:rhodef},  the following expression for the overall fixation probability:
\begin{equation*}
\rho=\frac{1}{2}-\frac{(e_{12}-e_{21})\left(\frac{B}{2}-d_1\right)}{B(e_{12}+e_{21})}.
\end{equation*}
We discover that $\rho=\frac{1}{2}$, and consequntly $K=u$, if (i) death rates are equal, $d_1 = d_2 = \frac{B}{2}$ (Result 1), or (ii) node-specific fixation probabilities are equal, $\rho_1= \rho_2 = \frac{1}{2}$ (Result 2).
When death rates are unequal and node-specific fixation probabilities are unequal, there are cases in which  $\rho<\frac{1}{2}$ ($K<u)$ and cases in which $\rho>\frac{1}{2}$ ($K>u$).
In particular, suppose that $d_1 > \frac{B}{2}$.  Then $\rho$ is less than $\frac{1}{2}$ if  $e_{12}<e_{21}$ and greater than $\frac{1}{2}$ if  $e_{12}>e_{21}$.

\subsubsection*{Population size $N=3$}

We solve Eqs.~\eqref{eq:rhorecur} and \eqref{eq:rhosum} to obtain expressions for $\rho_1$, $\rho_2$ and  $\rho_3$:
\begin{align}
\rho_{1}&=\frac{e_{12}e_{13}+e_{12}e_{23}+e_{13}e_{32}}{e_{12}\left(e_{13}+e_{23}+e_{31}\right)+e_{13}\left(e_{21}+e_{32}\right)+\left(e_{21}+e_{31}\right)\left(e_{23}+e_{32}\right)},\label{eq:rho1_3node}\\
\rho_{2}&=\frac{e_{21}e_{13}+e_{21}e_{23}+e_{23}e_{31}}{e_{12}\left(e_{13}+e_{23}+e_{31}\right)+e_{13}\left(e_{21}+e_{32}\right)+\left(e_{21}+e_{31}\right)\left(e_{23}+e_{32}\right)},\label{eq:rho2_3node}\\
\rho_{3}&=\frac{e_{12}e_{31}+e_{21}e_{32}+e_{31}e_{32}}{e_{12}\left(e_{13}+e_{23}+e_{31}\right)+e_{13}\left(e_{21}+e_{32}\right)+\left(e_{21}+e_{31}\right)\left(e_{23}+e_{32}\right)}.\label{eq:rho3_3node}
\end{align}
Substituting Eqs.~\eqref{eq:rho1_3node}-\eqref{eq:rho3_3node} into Eq.~\eqref{eq:rhodef} gives an expression for the overall fixation probability $\rho=\frac{\text{num}}{\text{denom}}$ where
\begin{align*}
\text{num} &= d_{1}
\left(e_{12}e_{13}+e_{12}e_{23}+e_{13}e_{32}\right)
+d_{2}\left(e_{21}e_{13}+e_{21}e_{23}+e_{23}e_{31}\right)\\
& \quad +d_{3}
\left(e_{12}e_{31}+e_{21}e_{32}+e_{31}e_{32}\right)\\
\text{denom} &=B\left( e_{12}\left(e_{13}+e_{23}+e_{31}\right)
+ e_{13}\left(e_{21}+e_{32}\right) + \left(e_{21}+e_{31}\right)\left(e_{23}+e_{32}\right)\right).
\end{align*}

We factor to obtain an expression for $\text{num}(\Delta \rho)$  in terms of $\tilde{\Delta} d_i = d_i-\frac{B}{3}$  for $i\in\{1,2,3\}$,
\begin{align}
\text{num}(\Delta \rho)
&=(e_{12}\tilde{\Delta} d_1+e_{21}\tilde{\Delta} d_2)(d_3- b_3)-e_{32}\tilde{\Delta} d_1(d_1- b_1)-e_{31}\tilde{\Delta} d_2(d_2- b_2)
\label{eq:rho12NK}
\end{align}

From  Eq.~\eqref{eq:rho12NK} we see that $\rho=\frac{1}{3}$, and conseqently $K=u$, in the case of equal death rates ($d_1=d_2=d_3=\frac{B}{3}$) and in the case of equivalent birth and death rate at each site ($b_i = d_i$ for $i \in \{1, 2, 3\}$).  This conforms to Result 1 and 2.  When death rates are nonuniform and node-specific probabilities are also nonuniform, we find cases in which $K>u$ and $K<u$.  For example, if $\tilde{\Delta} d_1 = 0, \tilde{\Delta} d_2 > 0, d_2 > b_2, $ and $d_3 < b_3$ then  Eq.~\eqref{eq:rho12NK} yields $\text{num}(\Delta \rho)<0$ and consequently, $\rho < \frac{1}{3}$ and $K<u$.  On the other hand, if $\tilde{\Delta} d_1 = 0, b_2 > d_2 > \frac{B}{3}, $ and $ d_3 > b_3$ then  Eq.~\eqref{eq:rho12NK} yields $\text{num}(\Delta \rho)>0$ and therefore, $\rho > \frac{1}{3}$ and $K>u$.

\subsection*{Upstream-downstream populations}

We now turn to the upstream-downstream model introduced in the Results and in Figure \ref{fig:updown}.

\begin{theorem}
In the upstream-downstream model, $\rho>1/N$, and consequently $K>u$, if and only if $d_\uparrow > d_\downarrow$.  
\end{theorem}

\begin{proof}
From Eq.~\eqref{eq:rhodef} we obtain the following expression for the overall fixation probability
\begin{align}
\rho &= \frac{1}{B}\left(N_{\uparrow} d_{\uparrow}\rho_{\uparrow}+ N_{\downarrow} d_{\downarrow}\rho_{\downarrow}\right)\nonumber\\
&=\frac{1}{N}+\frac{1}{B} \left[ N_{\uparrow} d_{\uparrow} \left(\rho_{\uparrow}-\frac{1}{N}\right)+ N_{\downarrow} d_{\downarrow} \left( \rho_{\downarrow}- \frac{1}{N} \right)\right]
\label{eq:rho_updown_expand}
\end{align}
Since fixation probabilities sum to one, and since the total population size is $N=N_\uparrow + N_\downarrow$, we have
\begin{equation}
\label{eq:Nrhorelation}
N_\downarrow \left( \rho_{\downarrow}- \frac{1}{N} \right) = - N_{\uparrow} \left(\rho_{\uparrow}-\frac{1}{N}\right).
\end{equation}
Substituting in Eq.~\eqref{eq:rho_updown_expand} yields
\begin{equation}
\label{eq:rho_updown_expand2}
\rho = \frac{1}{N}+(d_{\uparrow}-d_{\downarrow})\frac{N_{\uparrow}}{B}  \left(\rho_{\uparrow}-\frac{1}{N}\right).
\end{equation}

It follows from Eq.~\eqref{eq:rhoupdown} and from $e_\rightarrow>e_\leftarrow$ that $\rho_{\uparrow}>\rho_{\downarrow}$.  Moreover, Eq.~\eqref{eq:Nrhorelation} implies that $\rho_{\uparrow}-1/N$ and $\rho_{\downarrow}-1/N$ have opposite signs, and since $\rho_{\uparrow}>\rho_{\downarrow}$, it follows that $\rho_{\uparrow}-1/N$ must be positive. Thus the second term on the right-hand side of Eq.~\eqref{eq:rho_updown_expand2} has the sign of $d_{\uparrow}-d_{\downarrow}$.  We therefore conclude that $\rho>1/N$, and consequently the molecular clock is accelerated relative to the well-mixed case ($K>u$), if and only if $d_\uparrow > d_\downarrow$. \end{proof}

\section*{Acknowledgements}

The Foundational Questions in Evolutionary Biology initiative at Harvard University is sponsored by the John Templeton Foundation (RFP-12-02).

\end{document}